\documentclass[USenglish,oneside,twocolumn]{article}

\usepackage[utf8]{inputenc}
\usepackage[big]{dgruyter_NEW}

\usepackage{amsmath,amssymb,amsfonts}
\usepackage{amsthm}
\usepackage{upgreek}
\usepackage{graphicx}
\usepackage{textcomp}
\usepackage[linesnumbered,ruled,vlined]{algorithm2e}
\usepackage[rightcaption]{sidecap}
\usepackage{xcolor}
\usepackage{wrapfig}
\usepackage{hyperref}
\usepackage{comment}
\usepackage{subcaption}
\usepackage[capitalize, nameinlink]{cleveref}

\theoremstyle{definition}
\newtheorem{theorem}{Theorem}
\newtheorem{definition}{Definition}[section]
\newtheorem{corollary}{Corollary}[theorem]

\newcommand{\Var}{\mathrm{Var}}
\newcommand{\nb}{Na\"ive Bayes }
 
\DOI{foobar}

\cclogo{\includegraphics{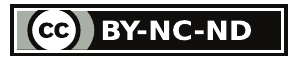}}
  
\begin{document}

  \author*[1]{Farzad Zafarani}

 \author[2]{Chris Clifton}

  \affil[1]{Department of Computer Science, Purdue University, USA., E-mail: farzad@purdue.edu}

  \affil[2]{Department of Computer Science, Purdue University, USA., E-mail: clifton@cs.purdue.edu}

  \title{\huge Differentially Private \nb Classifier Using Smooth Sensitivity}

  \runningtitle{Differentially Private \nb Classifier Using Smooth Sensitivity}


  \begin{abstract}
{
There is increasing awareness of the need to protect individual privacy in the training data used to develop machine learning models.
Differential Privacy is a strong concept of protecting individuals.  \nb is a popular machine learning algorithm, used as a baseline for many tasks. In this work, we have provided a differentially private \nb classifier that adds noise proportional to the \emph{smooth sensitivity} of its parameters. We compare our results to Vaidya, Shafiq, Basu, and Hong \cite{vaidya2013differentially} which scales noise to the global sensitivity of the parameters. Our experimental results on real-world datasets show that smooth sensitivity significantly improves accuracy while still guaranteeing $\varepsilon$-differential privacy.}
\end{abstract}

  \journalname{Proceedings on Privacy Enhancing Technologies}
\DOI{Editor to enter DOI}
  \startpage{1}
  \received{..}
  \revised{..}
  \accepted{..}

  \journalyear{..}
  \journalvolume{..}
  \journalissue{..}

\maketitle
\section{Introduction}
With the growth of user data across the internet, it has become more important to protect users' sensitive information. One solution to this problem is privacy-preserving data analysis, providing ability to share information while protecting users' data. Dwork, McSherry, Nissim, and Smith \cite{dwork2006calibrating} introduced differential privacy, providing a strong privacy guarantee for statistical data release. At a high level, Differential Privacy guarantees that the outcome of a differentially private algorithm would be similar no matter if a particular individual contributes personal data to the database or not.
There are several common approaches to differential privacy, including the Laplace Mechanism, which perturbs the parameters of the model with noise that is drawn from the Laplace distribution, scaled to the impact of a single individual on the result. The exponential Mechanism is another important mechanism to guarantee $(\varepsilon,\delta)-$differential privacy \cite{mcsherry2007mechanism}.

A model generated by a machine learning algorithm, when trained on a dataset, can reveal information about the training dataset. There are a series of recent works that guarantee that the output of a machine learning model satisfies differential privacy. These include differentially private Decision Trees \cite{jagannathan2009practical}, SVM \cite{rubinstein2009learning}, Deep Neural Networks \cite{abadi2016deep}, and Logistic Regression \cite{chaudhuri2009privacy}.
Differential privacy is particularly relevant for ensuring that machine learning models do not disclose individual information and even has the promise of improving generalization \cite{DPDataReuse}.
\nb is a baseline for many classification tasks.
Vaidya, Shafiq, Basu, and Hong \cite{vaidya2013differentially} provided a differentially private algorithm for the \nb classifier. They use the Laplace Mechanism to provide this guarantee based on computing the global sensitivity of the parameters. One of the main drawbacks of using global sensitivity is that the amount of the noise added to the output can be high if their could be a dataset where an individual would have a large impact on an outcome. Nissim, Raskhodnikova, and Smith \cite{nissim2007smooth} provided a general solution for this problem. In their paper they compute the \textit{Smooth Sensitivity} for a given function $f$ building on the definition of \textit{local sensitivity}. The local sensitivity of $f$ is the maximum amount of change in $f$ if we change a single element in a \emph{particular} dataset $x$.
It is obvious that the local sensitivity of a given function is not greater than its global sensitivity. Ideally, we would like to add noise proportional to the local sensitivity of $f$, but this does not satisfy the definition of differential privacy (the amount of noise needed reveals too much about the data), hence, in \cite{nissim2007smooth} they compute a $\beta-$smooth function which is the smallest upper bound for the local sensitivity that provides $\varepsilon$-differential privacy.  In this paper, we show how this approach can be used to provide a $\varepsilon$-differentially private algorithm for the \nb classifier based on the smooth sensitivity of the parameters of the model.

Bun and Steinke \cite{bun2019average}also  provide an algorithm for estimating the mean of a distribution using i.i.d. sample $x$ using smooth sensitivity. They first assume a crude bound on the $\mu \in [a,b]$, then they truncate the samples, i.e. they remove the largest $m$ samples and smallest $m$ samples from $x$, and compute the mean of the $n-2m$ samples. Finally they project the estimated mean to the range $[a,b]$. We compare our algorithm to Vaidya, Shafiq, Basu, and Hong \cite{vaidya2013differentially} and \nb using Bun and Steinke \cite{bun2019average}'s mean estimation algorithm.

\section{Preliminaries}
We first give a brief overview of the \nb classifier, then provide an overview of differential privacy. More specifically, we state the definitions for $\varepsilon-$differential privacy and smooth sensitivity. 

\subsection{\nb Classifier}\label{nb_section}
The \nb classifier is a family of probabilistic classifiers that uses Bayes' theorem and assumes independence between features. That is, it assumes that the value of a particular feature is unrelated to any other features.
The \nb classifier can handle an arbitrary number of independent variables, whether continuous or categorical, and classifies an instance to one of a finite number of classes.

To train the \nb model, a set of training examples with a corresponding target label is provided. The task is to assign a new class $c_{MAP}$ to an unseen instance $X = <X_1, X_2, \ldots, X_m>$.
Thus, the learning task would be that for each instance $X= <X_1,X_2,\ldots,X_m>$ consists of $m$ features, it assigns probability 
$$\Pr(c_j | X_1,X_2,\ldots,X_m)$$ for each of $K$ possible label classes $C = \{c_1, c_2, \ldots, c_K\}$.

$$c_{MAP} = \arg \max\limits_{j \in \{1, \ldots, K\}} \Pr(c_j | X_1, X_2, \ldots, X_m)$$
By using Bayes' theorem, we can further decompose the conditional probability to:
\begin{align*}
c_{MAP} &= \arg\max\limits_{j \in \{1, \ldots, K\}} \left ( \frac{\Pr(X_1, X_2, \ldots, X_m|c_j)\Pr(c_j)}{\Pr(X_1, X_2, \ldots, X_m)} \right ) \\
&= \arg\max\limits_{j \in \{1, \ldots, K\}} (\Pr(X_1, X_2, \ldots, X_m|c_j)\Pr(c_j))
\end{align*}

The \nb classifier makes the further simplifying assumption that the attribute values are conditionally independent, given the target value. Therefore:
$$c_{NB} = \arg \max\limits_{j \in \{1, \ldots, K\}} \Pr(c_j) \prod\limits_{i=1}^{m}  \Pr(X_i | c_j)$$
where $c_{NB}$ denotes the final class label for the instance $X = <X_1,X_2,\ldots,X_m>$. 

From the training dataset, we can pre-compute the conditional probabilities $\Pr(X_i | c_j)$. Also, $\Pr(c_j)$ can be computed by counting the number of items that are labeled $c_j$ in the training dataset. As with Vaidya, Shafiq, Basu, and Hong's work \cite{vaidya2013differentially}, we deal with both categorical and numerical attributes. The way that we estimate the probability is different for each class:
\begin{itemize}
    \item {\em{Categorical Value}}: For a categorical attribute $X_i$ with $J$ possible attribute values $a_1,a_2, \ldots,a_J$, the probability $\Pr(X_i = a_k | c_j) = \frac{\tau\{X_i = a_k \land C=c_j\}}{\tau\{C=c_j\}}$, 
    where the $\tau\{x\}$ operator returns the number of elements in the training set $D$ that satisfy property $x$. To prevent division by zero, we  use \textit{Laplace smoothing} which adds 1 to all counts.
    \item {\em{Numerical Value}}: For a numerical attribute $X_i$, one standard approach is to assume that for each possible discrete value $c_k$ of $C$, the distribution of each continuous $X_i$ is Gaussian, and is defined by a mean and standard deviation specific to $X_i$ and $c_k$ \cite{mitchell1997machine}.
    To train such a \nb classifier we must therefore estimate the mean and standard deviation of these Gaussians,
    $$\mu_{ik} = E[X_i|C=c_k]$$
    $$\sigma^2_{ik}=E[(X_i-\mu_{ik})^2|C=c_k]$$ ,
    for each numerical attribute $X_i$ and each possible value $c_k$ of $C$.
    
    If the numerical values are bounded, one can use the Truncated normal distribution and estimate its parameters. The probability density function of the Truncated normal distribution for $a \leq x \leq b$ is: 
    $$f(x; \mu, \sigma, a, b) = \frac{1}{\sigma}\frac{\phi(\frac{x-\mu}{\sigma})}{\phi(\frac{b-\mu}{\sigma}) - \phi(\frac{a-\mu}{\sigma})}$$
    Where
    $$\phi(\tau) = \frac{1}{\sqrt{2\pi}} exp(-\frac{1}{2}\tau^2)$$
    After estimating the values for mean and variance, the probability that an instance is of class $C_j$ can be directly computed from the density function.
    
\end{itemize}

\subsection{Differential Privacy}
Dwork,  McSherry, Nissim, and Smith \cite{dwork2006calibrating} defined the notion of differential privacy. At a high level, differential privacy guarantees that if your data is a part of a database from which we release information, then the released information will be similar if your data is a part of the database or not. That is, your data will have a negligible impact on the released information. Hence, no meaningful information can be inferred about individuals. The definitions below come from their work.

\begin{definition}{(Laplace Distribution) The probability density function (p.d.f) of the Laplace distribution $Lap(\mu, \lambda)$ is $f(x|\mu, \lambda) = \frac{1}{2\lambda} e^{-|x-\mu|/\lambda}$ with mean $\mu$ and standard deviation  $\sqrt{2}\lambda$.}
\end{definition}

\begin{definition} {(Cauchy Distribution) The probability density function (p.d.f) of the Cauchy distribution $Cauchy(x_0, \lambda)$ is $f(x|x_0, \lambda) = \frac{1}{\lambda\pi(1 + ((x-x_0)/\lambda)^2)}$ with location parameter $x_0$ and scale $\lambda$.}
\end{definition}

Let $\mathcal{D}_1, \ldots, \mathcal{D}_m$ denote domains, each of which could be categorical or numerical. A database $D$ consists of $n$ rows, $\{X^{(1)}, X^{(2)}, \ldots, X^{(n)}\}$, where each $X^{(i)}\in \mathcal{D}_1 \times \ldots \times \mathcal{D}_m$.

We say two databases $D_1$ and $D_2$ are at distance $k$ of each other and we write it as $d(D_1, D_2)=k$ if they differ by $k$ rows. Two database $D_1$ and $D_2$ are called \textit{neighbors} if $d(D_1, D_2)=1$.

\begin{definition}{
(Global Sensitivity). For $f: \mathcal{D} \rightarrow \mathcal{R}$, the global sensitivity of $f$ with respect to $\ell_1$ metric is:
$$GS_f = \max\limits_{x, y: d(x,y)=1} ||f(x) - f(y)||_1$$}
\end{definition}

\begin{definition}{
\label{def:dp}
(Differential Privacy). A randomized Mechanism $\mathcal{M}:\mathcal{D}\rightarrow\mathcal{R}$ with domain $\mathcal{D}$ and range $\mathcal{R}$ is $\varepsilon$-differentially private if for all $D_1, D_2 \in \mathcal{D}$ satisfying $d(D_1,D_2)=1$, and for all sets $\mathcal{S} \subseteq \mathcal{R}$ of possible outputs:
$$\Pr[\mathcal{M}(D_1) \in \mathcal{S}] \leq e^{\varepsilon} \Pr[\mathcal{M}(D_2) \in \mathcal{S}] $$}
\end{definition}
We also introduce \emph{concentrated differential privacy}; while we do not require this, our comparison with \cite{bun2019average} involves situations where \cite{bun2019average} satisfies concentrated differential privacy rather than \cref{def:dp}.
\begin{definition}\label{def:cdp}{
(Concentrated Differential Privacy). A randomized Mechanism $\mathcal{M}:\mathcal{D}\rightarrow\mathcal{R}$ with domain $\mathcal{D}$ and range $\mathcal{R}$ is $\frac{1}{2}\varepsilon^2$-CDP if for all $D_1, D_2 \in \mathcal{D}$ satisfying $d(D_1,D_2)=1$:
$$\forall \alpha > 1 \quad D_\alpha(\mathcal{M}||\mathcal{M}(x')) \leq \frac{1}{2}\varepsilon^2\alpha $$
where $D_\alpha$ is the R\'enyi divergence and $\alpha$ is a divergence parameter.}
\end{definition}
We also make use of a couple of properties of the way differentially private mechanisms combine.
\emph{Sequential composition} states that privacy loss
is additive:  If we take ``multiple looks'' at the data,
the privacy budget $\varepsilon$ expended is the sum of the privacy budgets $\varepsilon_i$ of each ``look''.
\begin{theorem}\label{thm_comp} (Sequential Composition \cite{mcsherry2009differentially}\cite{dwork2009differential}) Let $\mathcal{M}_1 : \mathcal{D} \rightarrow \mathcal{R}_1$ be an $\varepsilon_1$-differentially private algorithm, and let $\mathcal{M}_2 : \mathcal{D} \rightarrow \mathcal{R}_2$ be an $\varepsilon_2$-differentially private algorithm. Then their combination, defined to be $\mathcal{M}_{1,2} : \mathcal{D} \rightarrow \mathcal{R}_1 \times \mathcal{R}_2$ by the mapping: $\mathcal{M}_{1,2}(x) = (\mathcal{M}_1(x),\mathcal{M}_2(x))$ is $\varepsilon_1+\varepsilon_2$-differentially private.
\end{theorem}

Dwork et al. \cite{dwork2006calibrating} showed how to calibrate the noise to the global sensitivity of the function $f$ such that it satisfies $\varepsilon-$differential privacy . In their work, they have shown that the magnitude of the noise is proportional to $Lap(0, GS_f/\varepsilon)$. Intuitively, whenever we add noise proportional to the global sensitivity of $f$, we are adding noise proportional to the maximum magnitude of changes in $f$.
\section{Smooth Sensitivity}
One drawback of computing the global sensitivity of $f$ is that for many functions, there may be some possible datasets where changing one individual can make a dramatic change in the outcome.  For example, suppose we are computing median on a value ranging from $0..1$.  The dataset consisting of individuals with values 0, 0, and 1 has median 0, but by changing one individual the median goes to 1 - so the added noise must essentially obscure the entire result.

In practice, most databases will not have this property.
Nissim, Raskhodnikova, and Smith \cite{nissim2007smooth} showed that we can add noise based on the actual dataset we have rather than a worst-case dataset, and still satisfy $\varepsilon$-differential privacy.  We now outline their result.
 
\begin{definition}
(Local Sensitivity). For $f: \mathcal{D} \rightarrow \mathcal{R}$ and $x\in \mathcal{D}$ the local sensitivity of $f$ at $x$ (with respect to the $\ell_1$ metric) is:
$$LS_f(x) = \max\limits_{y:d(x,y)=1}||f(x) - f(y)||_1$$
\end{definition}

Note that $GS_f = \max\limits_{x}~LS_{f}(x)$.
We would like to be able to add noise proportional to local sensitivity. However, the local sensitivity may itself be high sensitivity, i.e., noise magnitude may compromise privacy. Nissim, Raskhodnikova, and Smith define a \textit{Smooth bound} that addresses this issue by looking not just at neighbors of the current dataset, but also their neighbors, etc.

\begin{definition} \label{def:smoothbound} 
(A Smooth bound). For $\beta > 0$, a function $S: \mathcal{D} \rightarrow \mathbb{R}^{+}$ is a $\beta-$smooth upper bound on the local sensitivity of $f$ if it satisfies the following requirements:
\begin{itemize}
\item \makebox[4cm]{$\forall x \in \mathcal{D}:$\hfill}  $S(x) \geq LS_f(x)$
\item \makebox[4cm]{$\forall x, y \in \mathcal{D}, d(x,y)=1 :$\hfill}  $S(x) \leq e^\beta S(y)$
\end{itemize}
\end{definition}

\begin{definition}
(Smooth Sensitivity). For $\beta > 0$, the $\beta$-smooth sensitivity of $f$ is:
$$S^*_{f, \beta}(x) = \max\limits_{y \in \mathcal{D}}(LS_{f}(y)\cdot e^{-\beta d(x,y)})$$
\end{definition}

The smooth sensitivity $S^*_{f, \beta}(x)$ is the smallest function that satisfies \cref{def:smoothbound}.

 Nissim et al. \cite{nissim2007smooth} showed that one could do much better than scaling the noise to the global sensitivity of $f$ by adding noise proportional to the \textit{Smooth sensitivity} of $f$ where it will give much higher output accuracy.

 \subsection{Computing Smooth Sensitivity}
We now describe how to compute the smooth sensitivity of a function.
\begin{definition}\label{def:smoothdist}
The sensitivity of $f$ at distance $k$ is:
$$A^{k}(x) = \max\limits_{y\in \mathcal{D} : d(x,y) \leq k} LS_f(y)$$
\end{definition}
We can express the smooth sensitivity of $f$ in terms of $A^k$ as follows:
$$S_{f,\varepsilon}^*(x) = \max_{k=0, 1, \ldots, n}e^{-k \varepsilon}(\max\limits_{y:d(x,y)=k}LS_{f}(y))$$
$$=\max\limits_{k=0,1\ldots, n}e^{-k \varepsilon}A^{(k)}(x)$$

\subsection{Calibrating Noise to the Smooth Sensitivity}
To release a function $f$ of the database $D$, the curator computes $f$ and publishes $\mathcal{M}(D) = f(D) + \lambda Z$ where $Z$ is a random variable drawn from a noise distribution, and $\lambda$ is the scaling parameter.

\begin{theorem}\label{thm:cauchy}
(Nissim et al. \cite{nissim2007smooth}) Let $f: \mathcal{D} \rightarrow \mathbb{R}$ be any real-valued function and let $S: \mathcal{D} \rightarrow \mathbb{R}$ be a $\beta$-smooth upper bound on the local sensitivity of $f$. Then we have:
\begin{quote}
    if $\beta \leq \frac{\varepsilon}{2(\gamma + 1)}$ and $\gamma > 1$, the algorithm $x \rightarrow f(x) + \frac{2(\gamma+1)S(x)}{\varepsilon}\eta$ where $\eta$ is sampled with distribution $h(z) \approx \frac{1}{1+|z|^\gamma}$, is $\varepsilon$-differentially private.
\end{quote}
\end{theorem}
Nissim et al. \cite{nissim2007smooth} showed that to scale the noise to the smooth sensitivity of $f$, it is sufficient to sample from an \textit{admissible noise distribution}, defined as follows.

For a subset $\mathcal{S}$ of $\mathbb{R}^d$, we write $\mathcal{S} + \xi$ for the set $\{z + \xi | z \in \mathcal{S} \}$, and $\varepsilon^\delta.\mathcal{S}$ for the set $\{\varepsilon^\lambda.z|z\in \mathcal{S} \}$. We also write $a \pm b$ for the interval $[a - b, a + b]$.
\begin{definition}{ (Nissim et al. \cite{nissim2007smooth})
A probability distribution on $\mathbb{R}^n$, given by a density function $h$, is $(\alpha, \beta)$-admissible (with respect to $\ell_1$), if for $\alpha = \alpha(\varepsilon, \delta)$, $\beta=\beta(\varepsilon, \delta)$, the following two conditions hold for all $\delta \in \mathbb{R}^n$ and $\lambda \in R$ satisfying $||\xi||_1 \leq \alpha$ and $|\lambda| \leq \beta$, and for all measurable subsets $\mathcal{S} \subset \mathbb{R}^n$:
\begin{itemize}
    \item Sliding Property: $$\Pr\limits_{Z\sim h}[Z \in \mathcal{S}] \leq e^{\frac{\varepsilon}{2}}.\Pr\limits_{Z \sim h}[Z\in \mathcal{S} + \xi] + \delta/2$$
    \item Dilation Property: $$\Pr\limits_{Z \sim h}[Z \in \mathcal{S}] \leq e^{\varepsilon/2}.\Pr\limits_{Z \sim h}[Z \in e^{\lambda}.\mathcal{S}] + \delta / 2$$
\end{itemize}
}
\end{definition}

\begin{theorem}\label{noise_theorem}
(Nissim et al. \cite{nissim2007smooth}) For any $\gamma > 1$, the distribution with density $h(z) \approx \frac{1}{1 + |z|^\gamma}$ is $(\frac{\varepsilon}{2(\gamma + 1)}, \frac{\varepsilon}{2(\gamma + 1)})-$admissible. Moreover, the $n-$dimensional product of independent copies of $h$ is $(\frac{\varepsilon}{2(\gamma + 1)}, \frac{\varepsilon}{2(\gamma  + 1)})$ admissible.
\end{theorem}

Nissim et al. show that a Cauchy distribution satisfies \cref{noise_theorem}.  They also show that approximate differential privacy can be satisfied under smooth sensitivity using noise from a gaussian or laplace distribution.  While we show only pure $(\varepsilon,0)$-differential privacy below, it is easily extended to approximate differential privacy; we show how this compares empirically in \cref{sec:experiments}.

\section{Differentially Private \nb}
While smooth sensitivity has been known for some time, it is often challenging to apply to practical problems.  Unlike global sensitivity, which requires only a worst-case analysis, to use \cref{thm:cauchy} with a naive application of \cref{def:smoothdist} is exponential in dataset size.  We now show how smooth sensitivity can be applied to create a differentially private \nb classifier. We first compute the \nb parameters (\cref{nb_section}). We then perturb the parameters with noise that preserves $\varepsilon$-differential privacy. As stated in \cref{nb_section}, a standard approach for fitting a machine learning model to a numerical attribute is to assume that the underlying distribution is Gaussian. We also assume that the numerical feature values are bounded. Hence, we start by computing the smooth sensitivity for estimating the parameters of the Truncated normal distribution.

Note that we only train using a subset of the data.  The way this subset is defined enables a smooth bound, as neighboring databases result in (at worst) training on a different subset rather than completely new data.
\begin{definition}(Trimmed Sample)
Let $X = {x_1, x_2, \ldots, x_n}$ denote the sample in sorted order and $m$ be a trimming parameter. The trimmed sequence of sample $X$ is:
$$x_{m+1}, x_{m+2}, \ldots, x_{n-m}$$
In other words, we draw a window of size $n-2m$ on the data.
 \label{trim}
\end{definition}

\subsection{Smooth Sensitivity of the Mean}
We start by computing the Smooth sensitivity of the mean of a dataset.
\begin{theorem}\label{t_mean}
Given a list of $n$ bounded real numbers $V = \{x_1, x_2, \ldots, x_n \}$ in the range of the interval $[L, U]$,  the Smooth sensitivity of the mean of $V$ can be computed in $\mathcal{O}(n^2)$.
\end{theorem}

\begin{proof}
Without loss of generality we assume that $V$ is in non-decreasing order. Now consider the set $V = \{x_1, x_2, \ldots, x_n \}$ with mean $\mu_V = \frac{x_1+x_2+\ldots+x_n}{n}$ and a set $V' = \{x_1', x_2', \ldots, x_n' \}$ with mean $\mu_{V'}=\frac{x'_1+x'_2+\ldots+x'_n}{n}$ where $d(V,V') = k$, i.e., it differs from $V$ by $k$ elements such that $|\mu_{V} - \mu_{V'}|$ is maximized. In the case that $\mu_{V'} > \mu_{V}$, it is easy to see that we have to replace $x_1, x_2, \ldots, x_k$ with $U$ (See \cref{fig:mean_figure}). Similarly, in the case that $\mu_{V} > \mu_{V'}$, we should replace $x_{n-k+1}, x_{n-k+2}, \ldots, x_n$ with $L$. Iterating through all possible choices of $1 \leq k \leq n$ and changing one last element to its extreme case ($L$ or $U$) would give us the smooth sensitivity.
\begin{figure}
    \centering
    \includegraphics[scale=0.5]{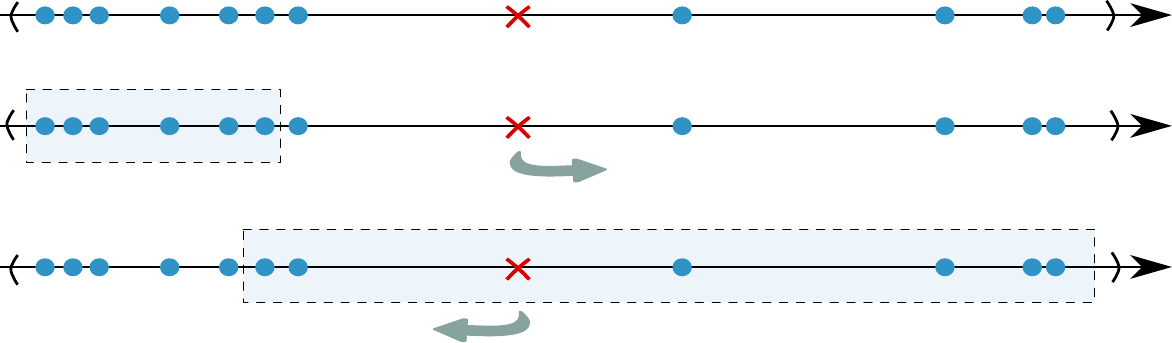}
    \caption{(top): blue dots represents points on the $x$ axis where the red cross mark represents their corresponding mean. (middle): shifting the mean towards the right-side by picking the $k$-smallest numbers and replacing them with $U$. (bottom): shifting the mean towards the left-side by picking the $k$-largest numbers and replacing them with $L$. }
    \label{fig:mean_figure}
\end{figure}
\end{proof}

\subsection{Smooth Sensitivity of Variance}
\label{sec:variance}
In this section, we will describe how to compute the smooth sensitivity of the variance of a dataset.
\begin{definition}{
($k$-maximal variance subset) Given a set $V = \{x_1, x_2, \ldots, x_n \}$ of $n$ real numbers in ascending order and an integer $k~(k < n)$, a subset $Q \subset V$ is called a $k$-maximal variance subset of $V$ if $|Q| = k$ and $\forall Q' \subset V, Q' \neq Q$ where $|Q'| = k$, $\Var[Q'] \leq \Var[Q]$.}
\end{definition}

\begin{theorem}
Given a list of $n$ bounded real numbers $V = \{x_1, x_2, \ldots, x_n \}$ in the range of interval $[L, U]$ and an integer $k~(k < n)$, the $k-$maximal variance subset can be computed in $\mathcal{O}(n^2)$.
\label{thm:maximal_var}
\end{theorem}

\begin{proof}
Without loss of generality assume that $V$ is in non-decreasing order. Let $Q = \{\Upgamma_1, \Upgamma_2, \ldots, \Upgamma_k\}$ be the $k-$maximal variance of $V$. We define the mean of $Q$ to be $\mu_Q = \frac{1}{k}(\Upgamma_1 + \Upgamma_2 + \ldots + \Upgamma_k)$. Let $Q' = \{\Upgamma_1, \Upgamma_2, \ldots, \Upgamma_k+\Updelta\}$, i.e., it only differs from $Q$ by adding $\Updelta$ to $\Upgamma_k$. Let $\mu_{Q'}$ be the mean of $Q'$, we have:
$$\mu_{Q'} = \mu_{Q} + \frac{1}{k} \Updelta$$
Let $\sigma^2_{Q}$ and $\sigma^2_{Q'}$ be the variance of $Q$ and $Q'$, respectively.
The variance of $Q = \{\Upgamma_1, \Upgamma_2, \ldots, \Upgamma_k\}$ (multiplied by $k$) is:
$$k\sigma^2_Q = (\Upgamma_1^2 + \ldots + \Upgamma_k^2) - k\mu^2_Q$$
Now we will see the impact of adding $\Updelta$ to $\Upgamma_k$ on the variance.
\begin{align*}
k\sigma_{Q'}^2 - k\sigma_{Q}^2 &= [(\Upgamma_k + \Updelta)^2 - \Upgamma_k^2] - k[\mu_{Q'}^2 - \mu_{Q}^2] \\
&= [\Upgamma_k^2 + \Updelta^2 + 2\Upgamma_k\Updelta - \Upgamma_k^2] - k(\mu_{Q} + \frac{\Updelta}{k})^2 \\
&\quad\quad\quad\quad\quad\quad\quad + k\mu_{Q}^2 \\
&= [\Upgamma_k^2 + \Updelta^2 + 2\Upgamma_k\Updelta - \Upgamma_k^2] - k\mu_{Q}^2 \\
&\quad\quad\quad\quad\quad\quad\quad - \frac{\Updelta^2}{k} - 2\mu_{Q} \Updelta + k\mu_{Q}^2 \\
&= 2\Updelta(\Upgamma_k - \mu_{Q}) + \frac{k-1}{k} \Updelta^2 \\
&\geq 2 \Updelta(\Upgamma_k - \mu_{Q}) 
\end{align*}

This value will be non-negative whenever the sign of $\Updelta$ and $(\Upgamma_k - \mu_{Q})$ are the same. That is, the difference in the variance will increase if we move $\Upgamma_k$ further from the mean $\mu_{Q}$. \\
Now assume that we are given a variance-maximizing sequence $S$ of $k$ values chosen from $V = \{x_1, x_2, \ldots, x_n \}$. Assume for the contradiction that $S$ contains an element $x_m$ where $x_m$ is not at the tail of $V$. That means there exists $x_r ~ (r \neq m)$, where if we replace $x_m$ with $x_r$ by the given inequalities we will increase the variance and $\Updelta  = x_m - x_r$ which contradicts that $S$ is a variance maximizing sequence.  So given $V = \{x_1, x_2, \ldots, x_n \}$, we know that the $k$-maximal variance subset selects elements from the tail of $V$. Iterating through all possible cases of $0 \leq i \leq k$, $j = k - i$, where the first $i$ elements are selected from the beginning of $V$, i.e., $(x_1, x_2, \ldots, x_i)$ , and $j$ elements are selected from the end of the sequence, i.e., $(x_{n-j+1}, \ldots, x_{n-1}, x_n)$, selecting the sequence that gives the maximum variance will be the solution to the $k$-maximal variance subset.
\end{proof}

\begin{corollary}\label{colremove}
Given a list of $n$ bounded real numbers $V = \{x_1, x_2, \ldots, x_n \}$ in the range $[L, U]$ and an integer $k~(k < n)$, the $k$-maximal variance subset can be achieved by removing $n-k$ consecutive elements in $V$.
\end{corollary}

\begin{definition}{
($k$-minimal variance subset) Given a set $V = \{x_1, x_2, \ldots, x_n \}$ of $n$ real numbers in ascending order and an integer $k~(k < n)$, a subset $Q \subset V$ is called $k$-minimal variance subset of $V$ if $|Q| = k$ and $\forall Q' \subset V, Q' \neq Q$ where $|Q| = k$, $\Var[Q'] \geq \Var[Q]$.}
\end{definition}

\begin{theorem}\label{thm:min_var}
Given a list of $n$ real numbers $V = \{x_1, x_2, \ldots, x_n \}$ in the range of the interval $[L, U]$ and an integer $k~(k < n)$, the $k$-minimal variance subset can be computed in $\mathcal{O}(n^2)$.
\end{theorem}
\begin{proof}
The proof is similar to the proof of \cref{thm:maximal_var}. Without loss of generality assume that $V$ is in non-decreasing order. Hence, given $V=\{x_1,x_2,\ldots,x_n\}$ the optimal solution would remove $\{x_1, x_2, \ldots, x_i\}$ and $\{x_{n - k + i}, \ldots, x_{n-1}, x_n\}$. Iterating through all possible $0 \leq i \leq k$ would give the optimal solution.
\end{proof}

\begin{theorem}\label{t_variance}
Given a list of $n$ bounded real numbers $V = \{x_1, x_2, \ldots, x_n \}$ in the range of interval $[L, U]$, the Smooth sensitivity of the variance of $V$ can be computed in $\mathcal{O}(n^2)$.
\end{theorem} 
\begin{proof}
\begin{figure}
    \centering
    \includegraphics[scale=0.6]{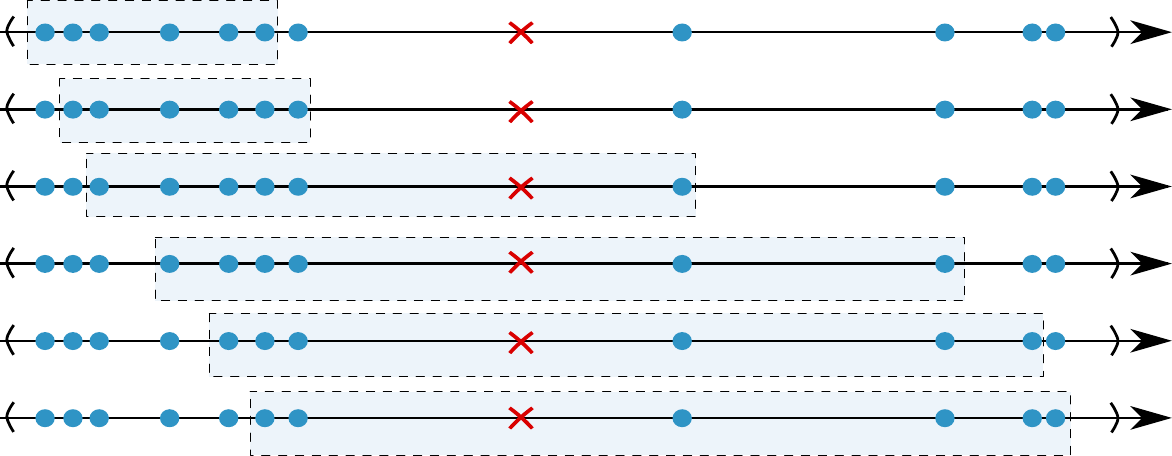}
    \caption{Blue dots represents points on the $x$ axis where the red cross mark represents their corresponding mean. For a fixed $k = 6$, we use a sliding window (shaded box) and move it through $x$ axis. We assign $t$ of the numbers to $L$ and $k-t$ to $U$ and save the maximum variance of the resulted sequence.}
    \label{fig:variance_figure}
\end{figure}
By \cref{colremove}, the $k$-maximal variance subset can be achieved by removing $n-k$ consecutive elements. Without loss of generality assume that $V$ is in non-decreasing order. We can iterate through all $k$ consecutive elements in $V$ and assign $t$ of them to be $L$, and $k-t$ of them to be $U$. The maximum over all possible cases would be the maximal variance. Similarly, for minimizing the variance by using \cref{thm:min_var}, we replace $\{x_1, x_2, \ldots, x_t \}$, i.e., the first $t$ elements and $\{x_{n - k + t}, \ldots, x_n \}$, i.e., the last $k-t$ elements with $\mu = \frac{x_{t+1}, x_{t+2}, \ldots, x_{n - k + t - 1}}{(n - k - 2)}$. Changing one element to its extreme case ($U$ or $L$), or to the mean of the sequence would give us the smooth sensitivity.
\end{proof}

\cref{examp_smooth} shows an example for the extreme change of the global sensitivity.

\begin{figure}[hb]
    \centering
    \includegraphics[scale=0.8]{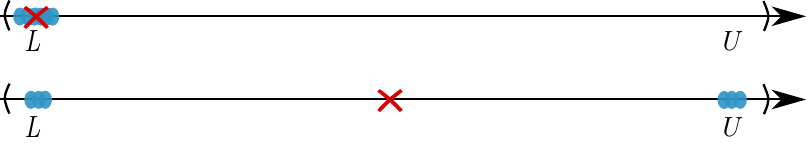}
    \caption{Maximum change of the global sensitivity.  (top): all points are on $L$, by changing $k=3$ points to the other extreme point ($U$), the variance and mean have the maximum change.}
    \label{examp_smooth}
\end{figure}


Note that another way to compute a differentially private variance of bounded dataset is to use $Var(X) = E[X^2]-E[X]^2$.  In this case, we would need to compute a differentially private $E[X]$ and $E[X^2]$.  While we already have the mean $E[X]$, we would need to use the privacy budget used for $Var(X)$ to instead calculate $E[X^2]$.  However, this results in $Var(X)$ being computed from two noisy values, resulting in a less accurate result than our method of computing it directly. 

\section{Bun and Steinke's Mean Estimation Using Smooth Sensitivity}
Bun and Steinke \cite{bun2019average} provided an algorithm for privately estimating the mean of a distribution. Their algorithm estimates the mean from an i.i.d. sample $x$ of the dataset. They first assume a crude bound on the mean $\mu \in [a,b]$, then they truncate the samples, i.e. they remove the largest $m$ samples and smallest $m$ samples from $x$, and compute the mean of the $n-2m$ samples. Finally they project the estimated mean to the range $[a,b]$. They have defined several admissible distributions such that adding noise proportional to them would guarantee $\frac{1}{2}\varepsilon^2$-CDP (Concentrated Differential Privacy). These include Student's T, Laplace log-normal, uniform log-normal, and arsinh-normal. \cref{thm:bunsteinke} is the main result of their paper:
\begin{theorem} (Bun and Steinke \cite{bun2019average})
Let $n \geq O(\log((b-a)\sigma)/\varepsilon)$, then there exist a $\varepsilon$-DP (or $\frac{1}{2}\varepsilon^2$-CDP) algorithm $\mathcal{M}: R^n \rightarrow R$ such that, for all $\mu \in [a,b]$, we have:
$$\mathbb{E}[(\mathcal{M}(x)-\mu)^2] \leq \frac{\sigma^2}{n} + \frac{\sigma^2}{n^2}.O(\frac{\log \frac{b-a}{\sigma}}{\varepsilon}+\frac{\log n}{\varepsilon^2})$$
\label{thm:bunsteinke} 
\end{theorem}
In \cref{thm:bunsteinke}, the first part is the non-private optimal mean-squared error and the additional term is the cost of the privacy.  The key difference between our approach and that of \cite{bun2019average} is that we assume structural or data-independent bounds on values (e.g., age between 0 and 125).  While less general than \cite{bun2019average}, in practical cases this allows for better differentially private estimates.  This also enables an independent private estimate of variance (\cref{sec:variance}) that provides better results than basing variance on estimates of the mean.  The approach of \cite{bun2019average} also requires choosing a smoothing parameter; it is not clear how to do this in a data independent or differentially private manner (although in our experiments and those of \cite{bun2019average}, the results were not that sensitive to the choice of smoothing parameter.)

Lemma 9 in \cite{bun2019average} shows that concentrated-differential privacy holds for the Laplace log-normal distribution.  From \cite{bun2019average} and the smooth sensitivity paper \cite{nissim2007smooth}, it follows that pure differential privacy holds (the above Theorem) when using the Cauchy distribution.  To give a fair comparison with the pure differential privacy of our approach, we use the Cauchy distribution with Bun and Steinke's approach in our comparisons.

\section{Algorithm}
We now give pseudocode to describe the algorithm for computing the smooth sensitivity of the \nb classifier. At a high level, we first compute the parameters of the \nb model, then compute the smooth sensitivity of each parameter and perturb the parameters with noise drawn from a Cauchy distribution. From \cref{t_mean} and \cref{t_variance}, one can compute the smooth sensitivity of the parameters for fitting a Gaussian distribution to the continuous data. For discrete variables the sensitivity is $1$ and can be perturbed by adding the small amount of noise $1/\varepsilon$.

We use an equal division of privacy budget between all accesses to data in keeping with \cite{vaidya2013differentially}.  Our goal for this paper is to show the value of smooth sensitivity, so we have kept with their division; we briefly discuss other allocations of privacy budget in \cref{sec:diffbudget}.

\begin{algorithm}
\DontPrintSemicolon
\KwIn{
    \begin{itemize}
        \item Labeled training data $D=\{(X^{(i)}, y_i)\}_{i=1}^{n}$
        \item $\varepsilon$ : the privacy parameter
        \item $Lap(\alpha, \beta)$: samples the Laplace distribution with mean $\alpha$ and scale $\beta$
        \item $Cauchy(\alpha, \beta)$: samples the Cauchy distribution with location parameter $\alpha$ and scale $\beta$
        \item Bound $[a_i,b_i]$ for each numerical attribute $X_i$
    \end{itemize} 
}
\textit{num\textunderscore count} $\leftarrow$ number of numerical attributes of $D$\;
\textit{cat\textunderscore count} $\leftarrow$ number of categorical attributes of $D$\;
$\varepsilon' \leftarrow \varepsilon / (2 * \textit{num\textunderscore count + cat\textunderscore count + 1)}$\;
\For{each attribute $X_j$} {
    \If{$X_j$ is categorical} {
    	sensitivity, s $\leftarrow$ 1\;
		scale factor, sf $\leftarrow$ $s/ \varepsilon'$\;
	    $\forall$ counts $\tau\{X_j=a_k \land C=c_l\}' = \tau\{X_j=a_k \land C=c_l\} + Laplace(0, sf)$\;
		Use $\tau\{X_j=a_k \land C=c_l\}'$ to compute $\Pr(a_k|c_l)$\;
    }
    \ElseIf{$X_j$ is numeric} {
        Trim the given sample based on \cref{trim}.\;

        Compute the Smooth Sensitivity, $s$ for mean $\mu_{jk}$ as per \cref{t_mean} with bound $[a_j, b_j]$\;
		scale factor, $sf \leftarrow \sqrt{2}s/\varepsilon'$\;
		$\mu_{jk}' \leftarrow \mu_{jk} + Cauchy(0, sf)$\;
		Compute the Smooth Sensitivity, $s$ for the standard deviation $\sigma_{jk}$ as per \cref{t_variance} with bound $[a_j, b_j]$\;
		scale factor, $sf \leftarrow \sqrt{2}s/\varepsilon'$\;
		$\sigma_{jk}' \leftarrow \sigma_{jk}$ + Cauchy(0, $sf$)\;
		Use $\mu_{jk}'$ and $\sigma_{jk}'$ to compute $\Pr(x_i|c_k)$\;
    }
    \For{each class $c_j$} {
    	count $\tau\{C=c_j\}' \leftarrow \tau\{C=c_j\} +$ Cauchy(0, $\sqrt{2}S/\varepsilon'$)\;
		Use $\tau\{C=c_j\}'$ to compute the prior $\Pr(c_j)$\;
    }
}
\caption{{\sc Differentially private \nb Classifier}}
\label{alg:privnaive}
\end{algorithm}

Similar to Vaidya et al.'s \cite{vaidya2013differentially} approach, when Cauchy noise is added, it is possible to make mean, standard deviation, counts, and class prior negative. To prevent this, we truncate the negative values to zero (a postprocessing step that does not impact the differential privacy guarantee.)

\subsection{Complete Privacy Guarantee}
\begin{theorem}
\cref{alg:privnaive} provides $\varepsilon$-differential privacy.
\end{theorem}
\begin{proof}
Each step $i$ of the \cref{alg:privnaive} is $\varepsilon'_i$-differentially private. By \cref{thm_comp}, the composition of finite number of $\varepsilon'_i$-differentially private algorithm is itself $(\sum\limits_{i} \varepsilon'_i)-$differentially private.
\end{proof}

\subsection{Runtime Analysis of Algorithm \ref{alg:privnaive}}
Bounding the smooth sensitivity for a given dataset does come at a cost:
\begin{theorem}
Algorithm \ref{alg:privnaive} computes differentially private \nb in $\mathcal{O}(nk+n^2)$.
\end{theorem}
\begin{proof}
The most time consuming part of computing differentially private \nb classifier is \cref{thm:maximal_var} that computes the smooth sensitivity of the variance in $\mathcal{O}(n^2)$ where $n$ is the number of rows in dataset. The rest of the computation can be done in linear time, hence, the pre-processing time is $\mathcal{O}(n^2)$. The running time of the standard \nb is $\mathcal{O}(nk)$, therefore, the total running time of differentially private \nb with the added pre-processing is $\mathcal{O}(nk+n^2)$.
\end{proof}

\section{Empirical Analysis}
\label{sec:experiments}
We have implemented the \nb classifier of Algorithm \ref{alg:privnaive} in Python.  We show a comparison with the results presented in \cite{vaidya2013differentially}, as well as a more realistic example that was not used in their work. We have also compared our algorithm to Bun and Steinke's \cite{bun2019average} private mean estimation. Since they did not specify an algorithm for computing the variance of a distribution, and our variance computation requires known bounds on values, we have naively used their mean estimation algorithm to estimate the variance using $Var(x) = E[X^2]-E[X]^2$. This demonstrates the practical improvements realized with smooth sensitivity.  We also show practical computational costs required to achieve this benefit.
\subsection{Datasets}
The datasets used in our experiment include those from the UCI repository \cite{Dua:2019} used in \cite{vaidya2013differentially}: Adult, Mushroom, Skin, Seed and Glass.
We also give results from a much more realistic dataset: the IPUMS USA: Version 8.0 Extract of 1940 Census for U.S. Census Bureau Disclosure Avoidance Research.

The UCI \textbf{Adult dataset} is drawn from 1994 census data of the United States. It consists of a 48K record subset drawn from a stratified sample of the U.S. population, the binary classification task is to predict if the income of an individual is less than or equal to 50K or not.

The UCI \textbf{Mushroom dataset} includes descriptions of hypothetical samples corresponding to 23 species of gilled mushrooms in the Agaricus and Lepiota Family. Each species is identified as definitely edible, definitely poisonous, or of unknown edibility and not recommended.

The UCI \textbf{Skin dataset \cite{bhatt2010skin}} is collected by randomly sampling B,G,R values from face images of various age groups (young, middle, and old), race groups (white, black, and asian), and genders obtained from the FERET  \cite{phillips1998feret} and PAL databases.

The UCI \textbf{Seed dataset} includes group comprised kernels belonging to three different varieties of wheat: Kama, Rosa and Canadian, 70 elements each, randomly selected for
the experiment.

The UCI \textbf{Glass dataset} contains the description of 214 fragments of glass originally collected for a study in the context of criminal investigation. Each fragment has a measured reflectivity index and chemical composition (weight percent of Na, Mg, Al, Si, K, Ca, Ba and Fe).
    
The IPUMS \textbf{1940 Census} dataset is a sample drawn uniformly at random from the U.S. population, taken from the 1940 Census.
The Adult dataset, and most other IPUMS microdata sets, are stratified datasets intended to be used with weighted values, and as such are not representative of real populations when used as unweighted values. Using with weighted values poses additional difficulties for computing sensitivity that are beyond the scope of this paper. When used as training data for machine learning, they are not representative of performance on real populations.
The 1940 Census data is a uniform sample of the population, and as such is an appropriate representation for a machine learning task.\footnote{The IPUMS data is available at \url{https://usa.ipums.org/usa/1940CensusDASTestData.shtml .}}
We used the 13 attributes that are included in the adult dataset, and construct a binary classification task to predict whether the income of an individual is less than or equal to the mean income of the population, a similar prediction task to that used with the Adult dataset (although with a different threshold, due to inflation between 1940 and 1994.)  We discarded individuals with unknown values.  To give an idea of the variance across different subpopulations, we report values for different U.S. States.
\cref{dataset_info} shows the detailed description of the datasets that we have used for this experiment. 
\begin{table}[ht]
\centering
\begin{tabular}{lccc}
\hline
Dataset  & \multicolumn{1}{l}{No. of Records} & \multicolumn{1}{l}{Attributes} & \multicolumn{1}{l}{Classes} \\ \hline
Adult    & 48K                                & 14                        & 2                         \\
Mushroom & 8K                                 & 22                        & 2                         \\
Seed     & 210                                & 7                         & 3                         \\
Skin     & 245K                               & 3                         & 2                         \\
Glass     & 214                               & 9                         & 7                         \\
Wyoming  & 250K                               & 13                        & 2                         \\
Nevada   & 110K                               & 13                        & 2                         \\
Washington   & 1.7M                                 & 13                        & 2                         \\ 
Oregon   & 1M                                 & 13                        & 2                         \\ \hline
\end{tabular}
\caption{Description of the datasets used for our experiment.}
\label{dataset_info}
\end{table}

\subsection{Experimental Results}
Since there is randomness in our algorithm for adding noise, we have run all algorithms five iterations with 10-fold cross-validation; we show mean and error bars across the iterations. We use two baselines:  a standard \nb classifier and the constant ``predict the majority class'' classifier.  Probably the most interesting comparison is with the Differentially Private classifier by Vaidya, Shafiq, Basu, and Hong \cite{vaidya2013differentially} and Bun and Steinke's mean estimation algorithm \cite{bun2019average} as this shows the specific gains achieved through using smooth sensitivity rather than global sensitivity. 

\begin{figure}
\begin{minipage}{.48\textwidth}
\centering
	\includegraphics[width=\linewidth, height=6cm]{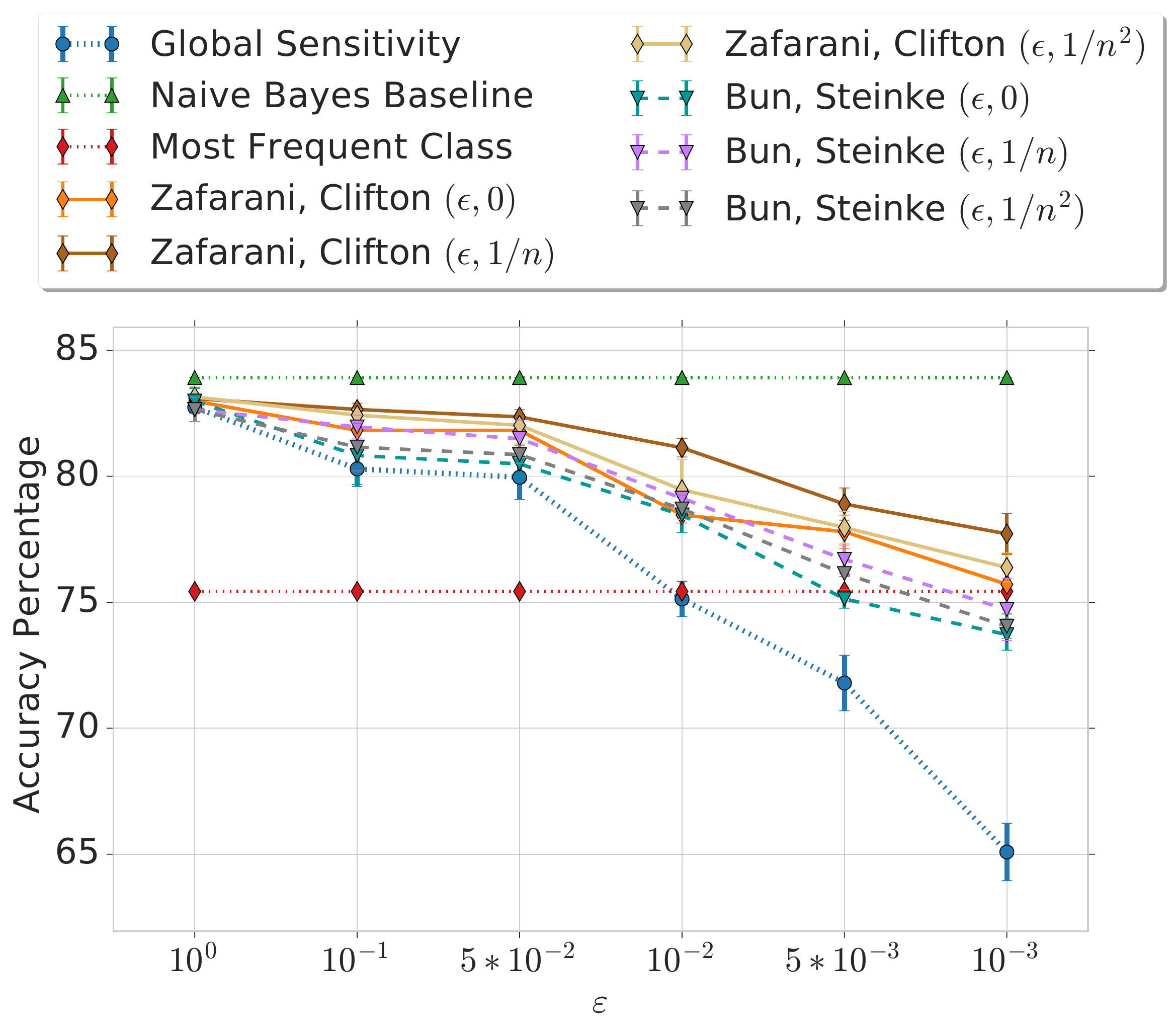}
	\caption{Accuracy at various values of $\varepsilon$: Adult dataset}
	\label{fig:adult}
\end{minipage}
    \par\bigskip
\begin{minipage}{.48\textwidth}

\centering
\includegraphics[width=\linewidth, height=6cm]{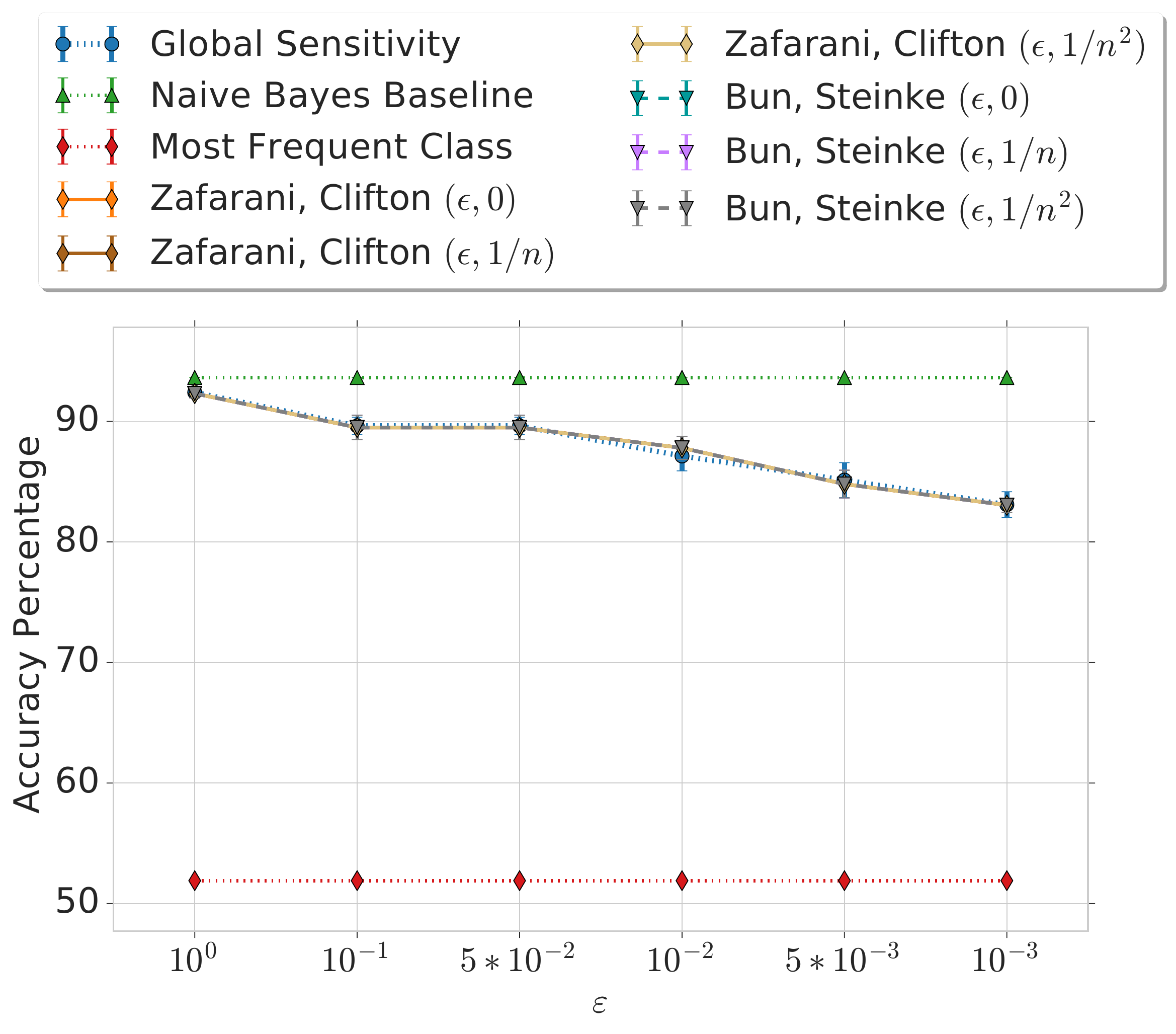}
    \caption{Accuracy at various values of $\varepsilon$: Mushroom dataset}
    \label{fig:mushroom}
\end{minipage}
\par\bigskip
\begin{minipage}{.48\textwidth}

    \centering
    \includegraphics[width=\linewidth, height=6cm]{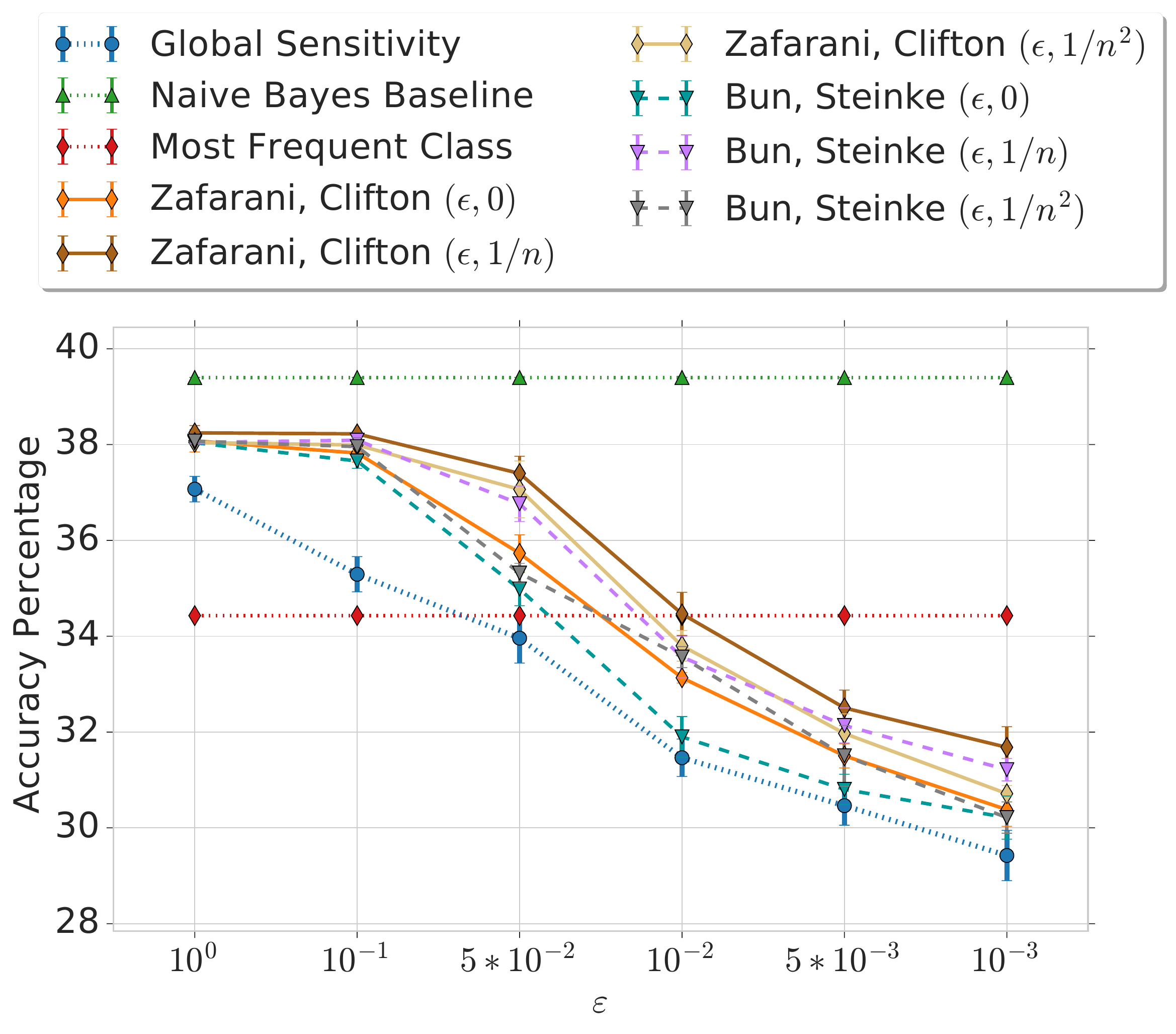}
    \caption{Accuracy at various values of $\varepsilon$: Glass dataset}
    \label{fig:glass}
    \end{minipage}
\end{figure}

Note that \cite{vaidya2013differentially} reports results in terms of the privacy budget used for each attribute; we instead report the total privacy budget utilized under sequential composition (\cref{thm_comp}).  This is simply a scaling of $\varepsilon$ and does not fundamentally change their reported results. 

Results are shown in Figures \ref{fig:adult}-\ref{fig:washington}.  We give the mean value across multiple draws from the differentially privacy mechanisms, as well as standard deviation.
\begin{figure}
\begin{minipage}{.48\textwidth}
\centering
    \includegraphics[width=\linewidth, height=6cm]{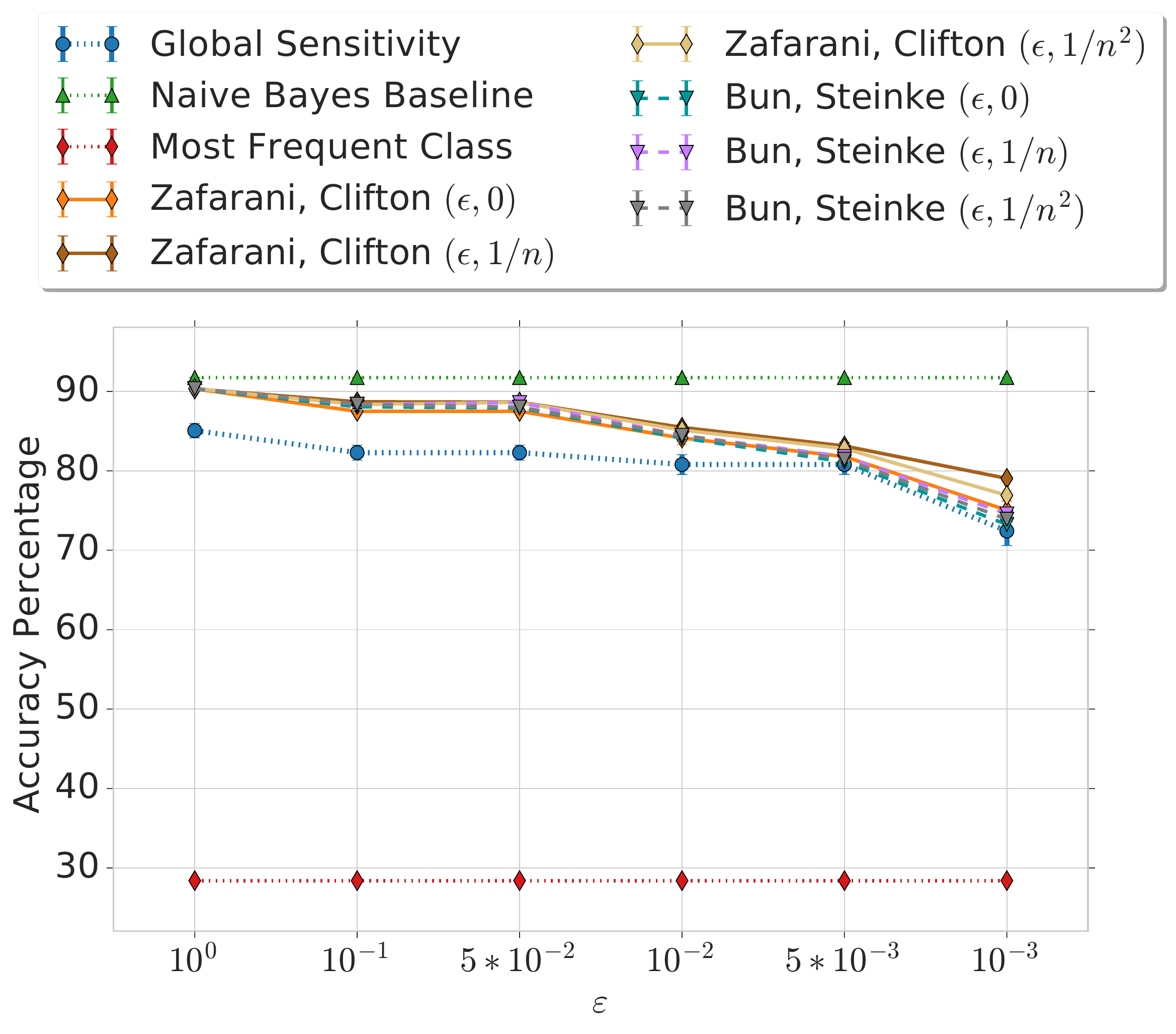}
    \caption{Accuracy at various values of $\varepsilon$: Seed dataset}
    \label{fig:seed}
 \end{minipage}
    \par\bigskip
 \begin{minipage}{.48\textwidth}
 \centering
    \includegraphics[width=\linewidth, height=6cm]{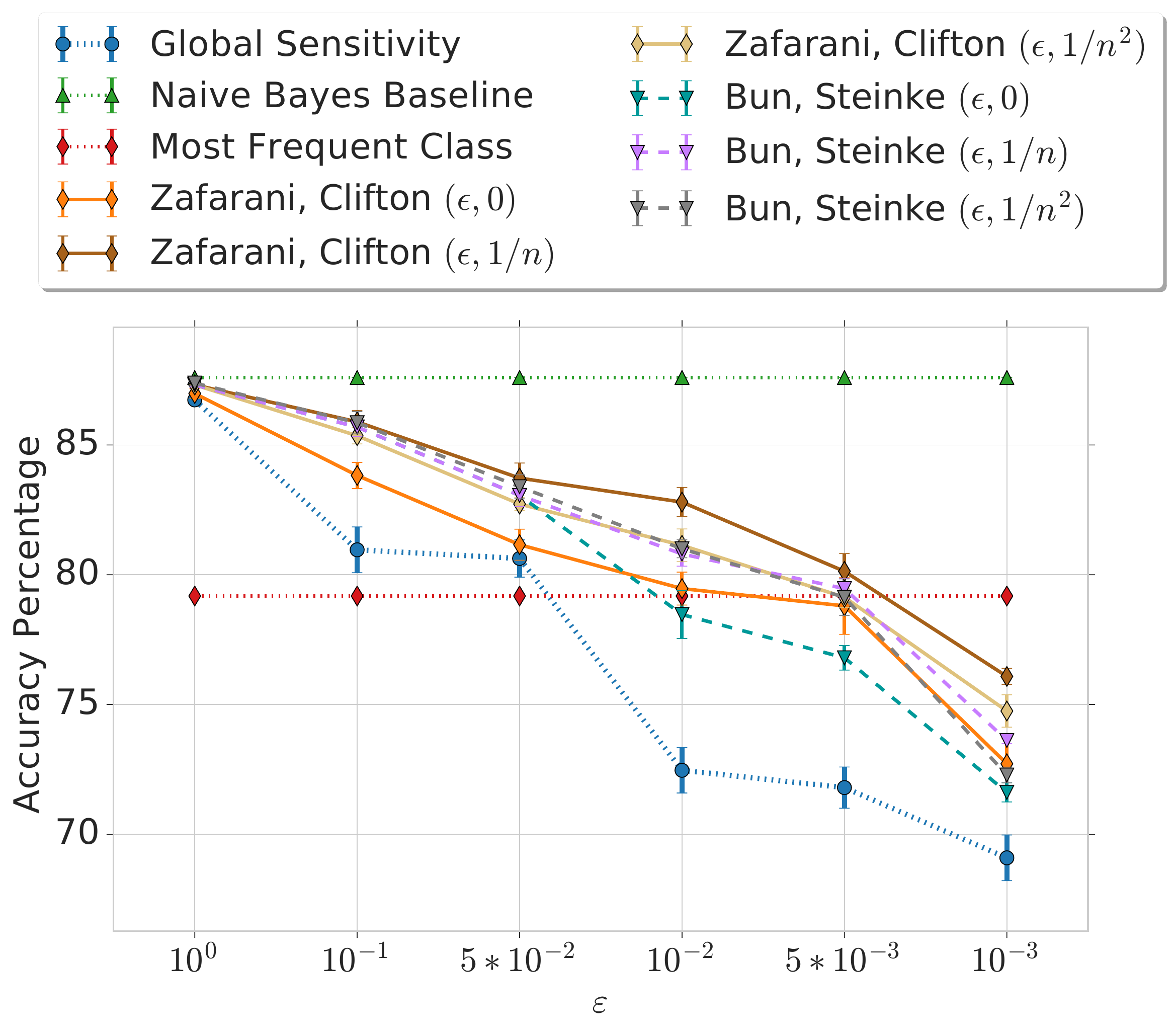}
    \caption{Accuracy at various values of $\varepsilon$: Skin dataset}
    \label{fig:skin}
 \end{minipage}
    \par\bigskip
 \begin{minipage}{.48\textwidth}
 \centering
    \includegraphics[width=\linewidth, height=6cm]{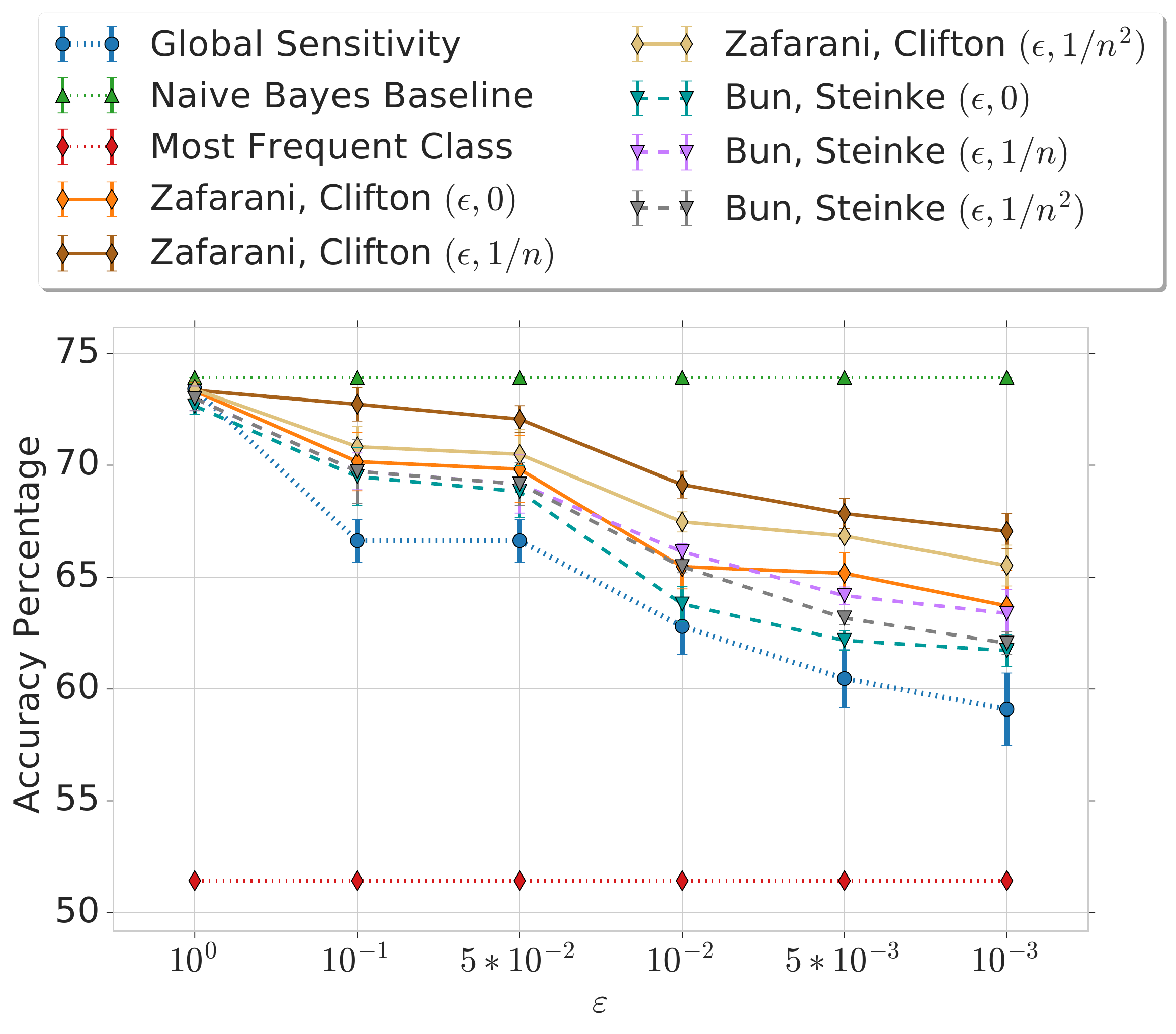}
    \caption{Accuracy at various values of $\varepsilon$: Wyoming from 1940 Census}
    \label{fig:wyoming}
\end{minipage}
\end{figure}
\begin{figure}
\centering
\begin{minipage}{.48\textwidth}
    \includegraphics[width=\linewidth, height=6cm]{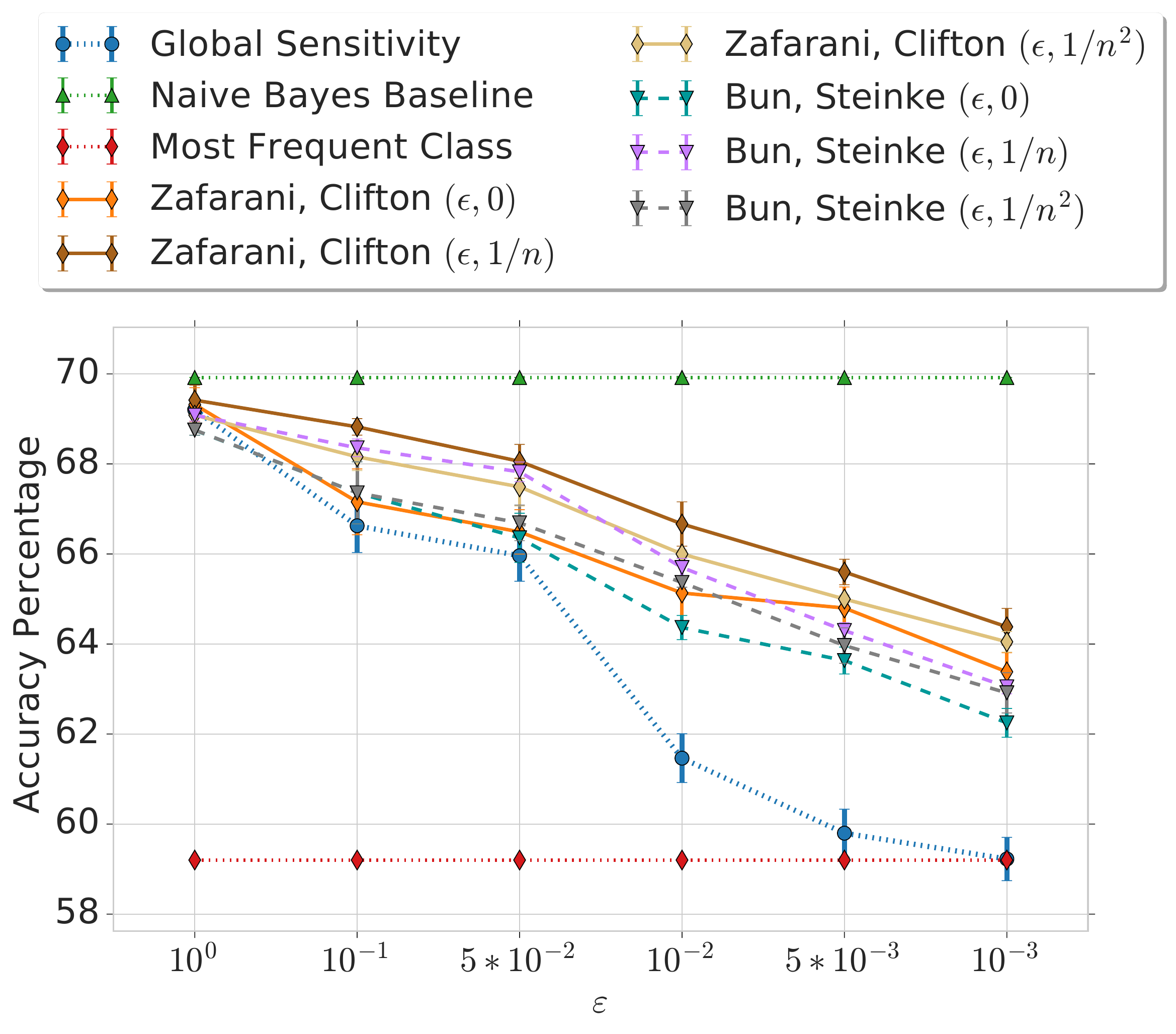}
    \caption{Accuracy at various values of $\varepsilon$: Nevada from 1940 Census}
    \label{fig:nevada}
    \end{minipage}
    \par\bigskip
    \begin{minipage}{.48\textwidth}
   \centering
    \includegraphics[width=\linewidth, height=6cm]{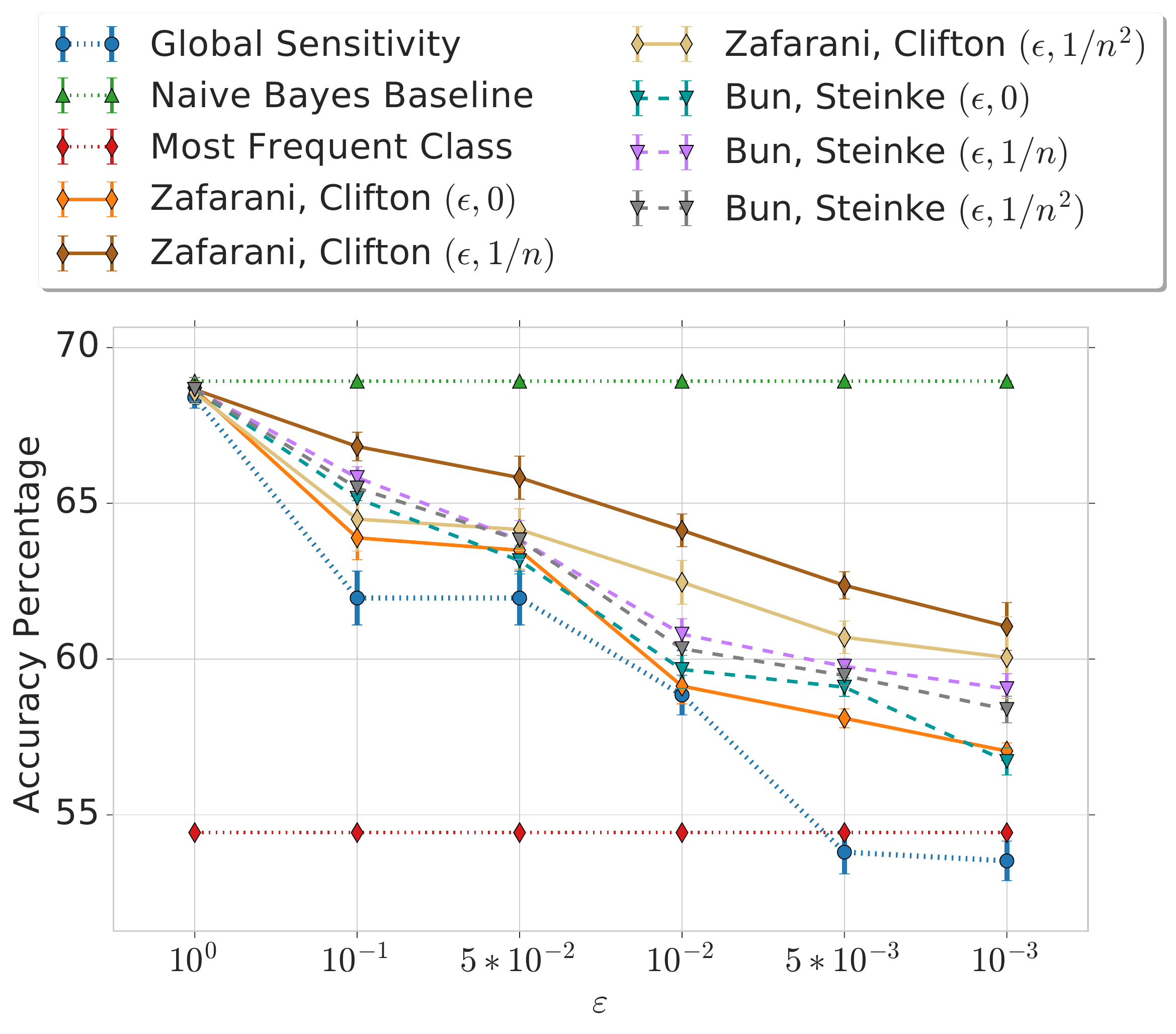}
    \caption{Accuracy at various values of $\varepsilon$: Oregon from 1940 Census}
    \label{fig:oregon}
\end{minipage}
    \par\bigskip
\begin{minipage}{.48\textwidth}
   \includegraphics[width=\linewidth, height=6cm]{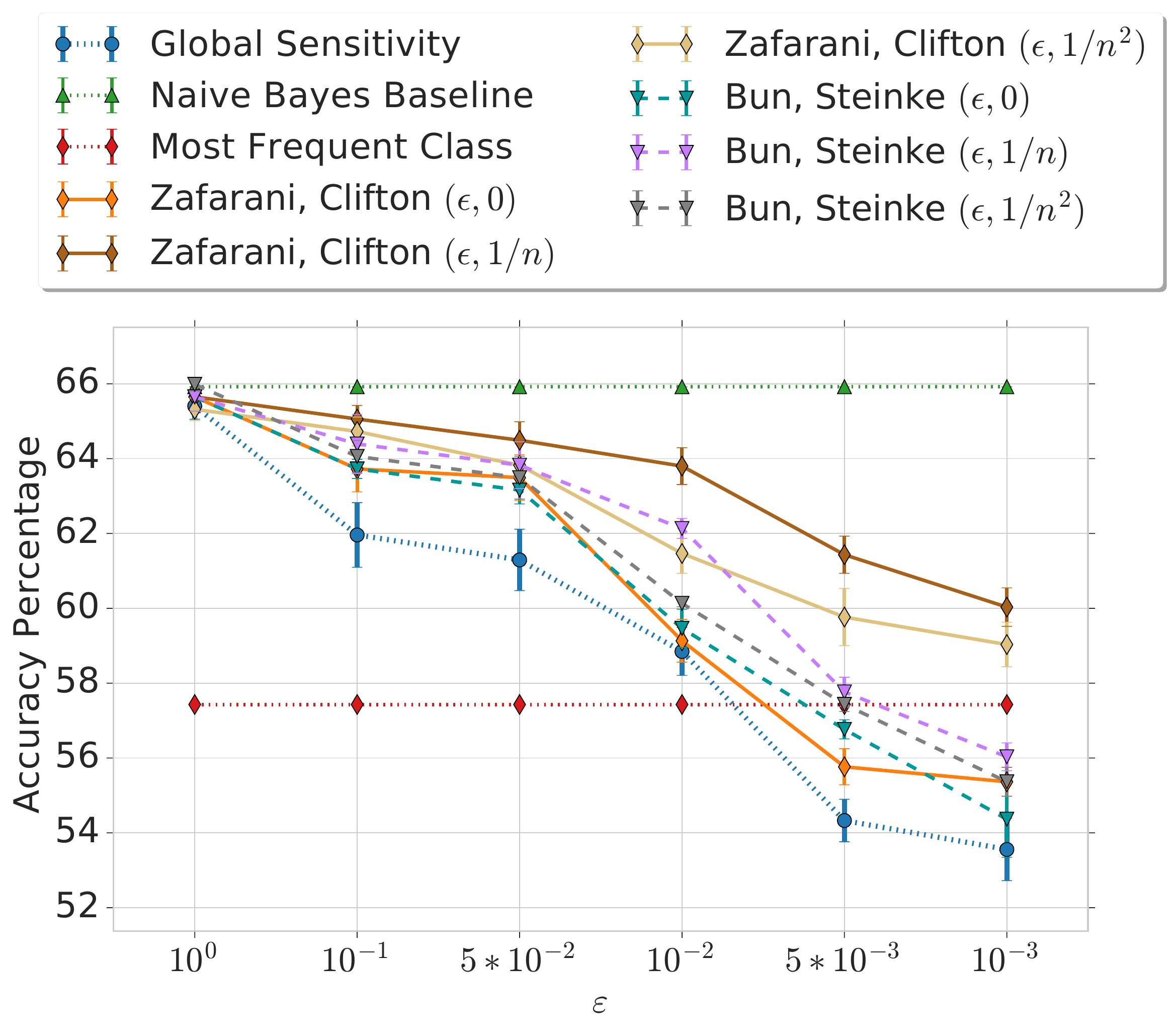}
    \caption{Accuracy at various values of $\varepsilon$: Washington from 1940 Census}
    \label{fig:washington}
\end{minipage}
\end{figure}
Smooth sensitivity gives significant improvements in classifier accuracy, particularly for the datasets using human data.
For completeness we also show results for $(\varepsilon, \delta)$-differential privacy \cite{mcsherry2009differentially}, a weaker form of differential privacy, for $\delta=1/n$ and $1/n^2$.  (For clarity, algorithm \ref{alg:privnaive} shows only $(\varepsilon,0)$-differential privacy; the extension to $(\varepsilon,\delta)$ is straightforward.)
As shown by Nissim et al. \cite{nissim2007smooth}, if we add Laplace or Gaussian noise of magnitude calibrated to the smooth sensitivity, this would give approximate differential privacy. In this work, to achieve approximate differential privacy, instead of using Cauchy distribution in \cref{alg:privnaive}, we have used Gaussian distribution.
As it is shown in Figures \ref{fig:adult}-\ref{fig:washington}, approximate-DP does give some accuracy improvement, but at a cost of a weaker privacy guarantee.

Since we only use smooth sensitivity with continuous values and the Mushroom dataset contains only categorical values, the results in Figure \ref{fig:mushroom} are the same for all methods; we include for completeness with \cite{vaidya2013differentially} and to show variance across parameter values.

\subsection{Computational Costs}\label{sec:compcost}
As we showed in Theorem \ref{thm_comp}, smooth sensitivity does not come for free.  The global sensitivity algorithm is $O(n \log n)$, smooth sensitivity takes us to $O(n^2)$.  To demonstrate what this means in practical terms, we show how results change as we vary the dataset size.  \cref{fig:runtime} compares the two approaches, showing
runtimes of Python implementations on a computer with an 2.2 GHz Intel Core i7 CPU, 64 GB 1600 MHz DDR3 RAM, running Ubuntu 18.04. While we do see a substantial runtime cost in training the \nb classifier, these are reasonable times for many practical applications.  It is interesting to note that as the dataset size increases, factors other than sensitivity calculation become increasingly important; the nearly order of magnitude difference in cost with a 5000 instance dataset drops to half an order of magnitude (although still a substantial time difference) with over a million instances.

\begin{figure}
    \centering
    \includegraphics[width=\linewidth, height=6cm]{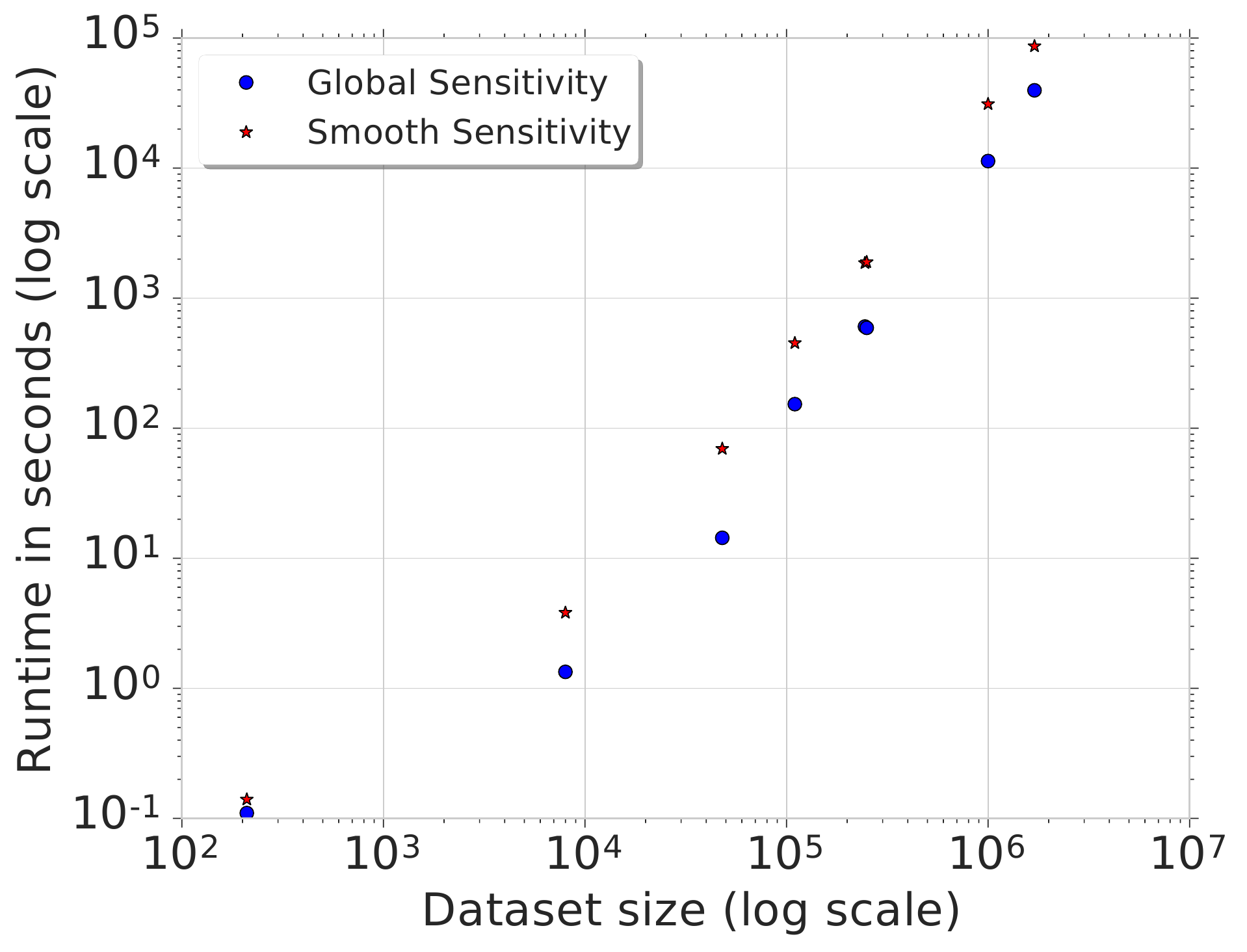}
    \caption{Runtime Analysis on datasets provided in Table \ref{dataset_info}, blue dot represents the Global Sensitivity training time while the red star represent the corresponding Smooth sensitivity training time in seconds ($\log \log$ scale).}
    \label{fig:runtime}
\end{figure}

\section{Related Work}
We have discussed the work of Vaidya et al. \cite{vaidya2013differentially},
which addresses the same problem we do using global sensitivity.  Li et al. \cite{li2018differentially} proposed a new model for a differentially private \nb classifier over multiple data sources. Their proposed method enables a trainer to train a \nb classifier over the dataset provided jointly by different data owners, without requiring a trusted aggregator as in our work and \cite{vaidya2013differentially}. Yilmaz et al. \cite{yilmaz2019locally} provided a differentially private \nb classifier under the local differential privacy setting. With local differential privacy, individuals perturb their data before sending to an \emph{untrusted} aggregator.
The stronger adversary model of these two approaches (eliminating the trusted aggregator) results in significantly more noise and reduced accuracy.

There have been studies of privacy-preserving \nb under different privacy models. Kantarcioglu et al. \cite{kantarcioglu2003privacy} proposed a privacy-preserving \nb classifier for horizontally partitioned data. Their solution uses secure summation and logarithm to learn a distributed \nb classifier securely. Vaidya and Clifton \cite{vaidya2004privacy} gave a solution to the same problem but for vertically partitioned data under the semi-honest model.  These approaches protect the data during training, an orthogonal problem to differential privacy's protection of disclosure via the learned model.

Making machine learning models differentially private has had more general interest, with solutions proposed for several machine learning approaches. Some of the more well known include Jagannathan, Pillaipakkamnatt and Wright \cite{jagannathan2009practical} that gives a differentially private algorithm for random decision trees, and Abadi et al.  \cite{abadi2016deep} that give a differentially private framework for deep learning models.  One possible area for future work is to determine if these approaches could benefit from using smooth sensitivity rather than global sensitivity, and if so, how that might be tractably computed.

Beyond machine learning,
there are many differentially private algorithms for statistical tests. Campbell et al. \cite{campbell2018differentially} gave a differentially private ANOVA test. Task and Clifton \cite{task2016differentially} provided differentially private significance testing on paired-sample data. 
In \cref{fig:diffbudget}, we show that a non-uniform allocation of privacy budget to different attributes can have an impact; we discuss this further in the Appendix.  This is not a trivial problem in the context of differential privacy.
Anandan and Clifton \cite{anandan2018differentially} provided a differentially private solution for a more basic version of this problem, feature selection for data mining tasks. In their work they analyze the sensitivity of various feature selection techniques used in data mining and show that some of them are not suitable for differentially private analysis due to high sensitivity.

\section{Conclusion}
In this paper, we have developed a differentially private \nb Classifier using Smooth Sensitivity for numerical data, along with global sensitivity for categorical values. For fitting numerical values, we have made the assumption typically used with \nb that the data follows a Gaussian distribution. When the features are bounded, we assume that the underlying data follows a truncated normal distribution. We have computed the smooth sensitivity of the parameters of the Gaussian, $\mu$ and $\sigma$. To obtain the $\varepsilon$-differential private algorithm, we have added noise proportional to the smooth sensitivity of the parameters. Previous work on \nb differential private classifier done by Vaidya, Shafiq, Basu, and Hong \cite{vaidya2013differentially} perturb the parameters of the \nb classifier by a noise that is scaled to the global sensitivity of the parameters. We demonstrate on real-world datasets that our method achieves a significant accuracy improvement.
While this comes at a computational cost, it is a cost only in model training. 
The released model is essentially identical to \cite{vaidya2013differentially}, but with higher accuracy, and still satisfying the $\varepsilon$-differential privacy definition.
Smooth sensitivity provides the same differential privacy guarantee as global sensitivity for given values of $\varepsilon$ and $\delta$, including $(\varepsilon, 0)$.  The greater the impact of numerical features on the result, the greater the benefit of smooth sensitivity.
We have also compared our result with private mean estimation of Bun and Steinke \cite{bun2019average} where they estimate the mean of an unknown distribution based on an i.i.d. sample from the dataset, adding noise proportional to the smooth sensitivity of the truncated mean.

\section*{Acknowledgements}
This research received no specific grant from any funding agency in the public, commercial, or not-for-profit sectors.

\bibliographystyle{ieeetr}
\bibliography{main}

\section{Appendix}
\subsection{Varying Allocation of the Privacy Budget}\label{sec:diffbudget}
\begin{figure}
    \includegraphics[width=\linewidth, height=6cm]{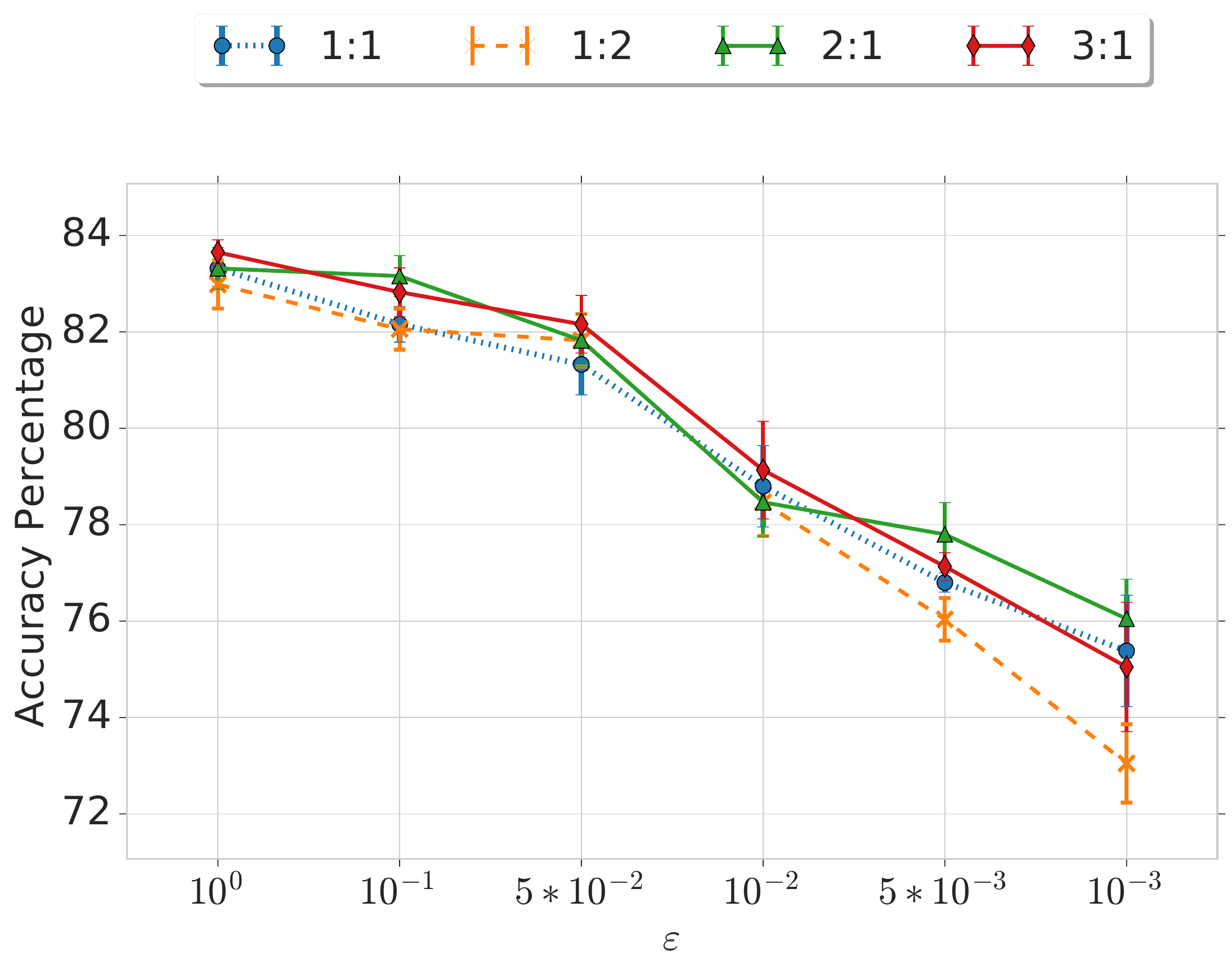}
    \caption{Adult dataset: Accuracy at various values of $\varepsilon$ for varying allocation of privacy budget - numerical:categorical}
    \label{fig:diffbudget}
\end{figure}

We briefly experimented with varying the allocation of the privacy budget between categorical and numeric variables, an example is shown in \cref{fig:diffbudget}.  With the exception of substantially overweighting categorical variables and giving little privacy budget (high noise) for numeric (line 1:2), the results are inconclusive; the differences are well within the error seen across cross-validation groups.
This suggests that numeric attributes are important to this classification problem, and due to the higher sensitivity, giving too little privacy budget to them limits their value.

\cref{alg:privnaive} and all other figures correspond to the ``2:1'' allocation in  \cref{fig:diffbudget}:  $2*\varepsilon'$ allocated to each numeric attribute, one for the mean estimation, and one for standard deviation.

Further research would be needed to establish reasoning to apply additional weight to numeric or categorical variables.  Making the selection empirically before applying differential privacy would constitute a disclosure of information, and those violate the provided $\varepsilon-$differential privacy guarantee.
(Our choice of a ``2:1'' allocation outside of \cref{fig:diffbudget} was based on the need to need to gather both mean and standard deviation, rather than based on empirical analysis.)
Properly arriving at a weighting in a way that is both private and effective is an interesting problem, and one that comes up generally in differentially private machine learning.

To measure the effect of privacy budget on numerical and categorical attributes we further experimented on a synthetic dataset. We generated two synthetic datasets of 10000 entries with 5 categorical attributes and 5 numerical. In the first dataset, the numerical attributes are correlated with the label while the categorical values are uncorrelated. In the second dataset, the categorical attributes are correlated with the label while numerical attributes are uncorrelated. The impact of different distributions of privacy budget are presented in  \cref{num-cor} and \cref{cat-cor}.  Unsurprisingly, the optimal privacy budget distribution is dataset dependent.

\begin{figure}
    \centering
    \includegraphics[width=\linewidth, height=6cm]{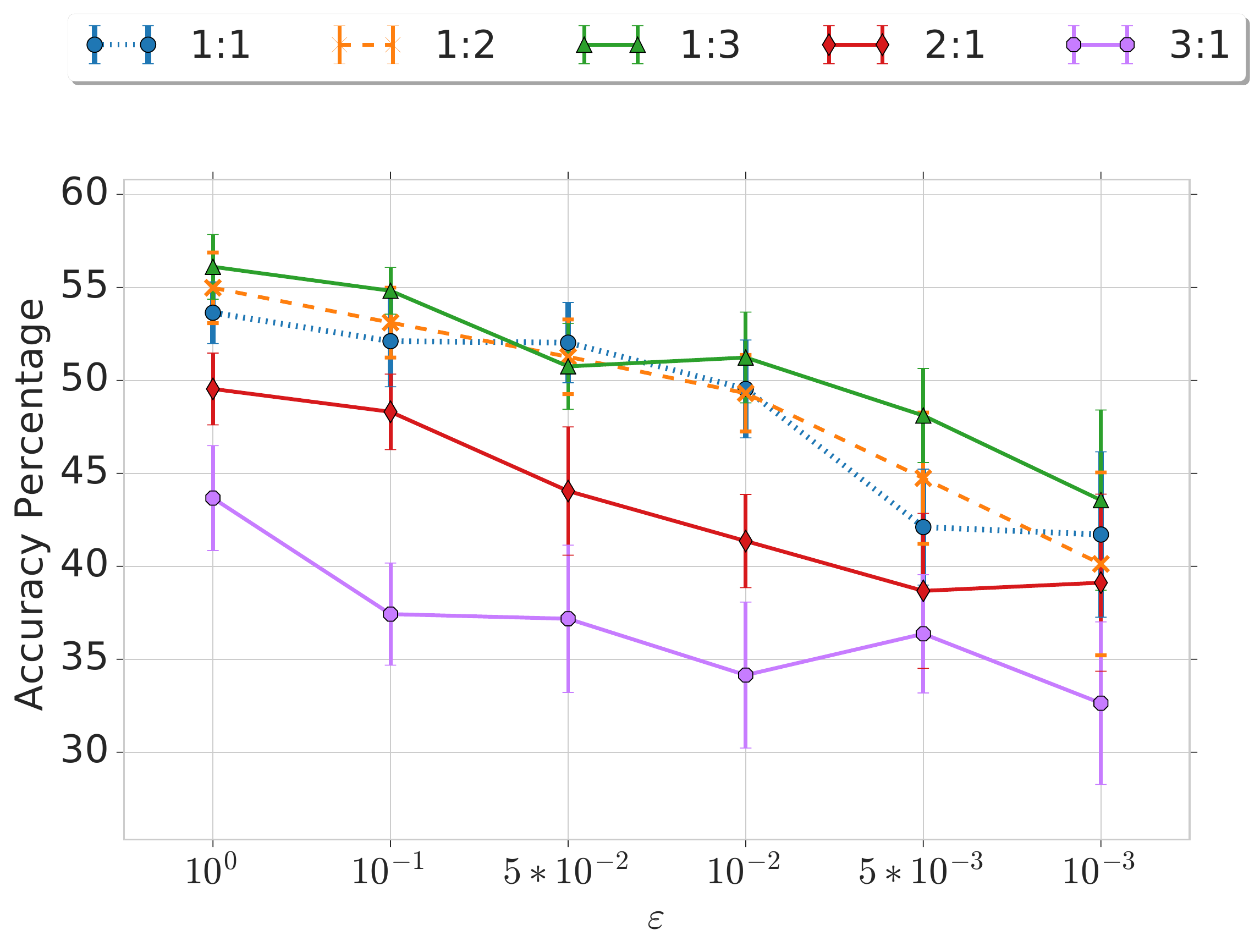}
    \caption{Accuracy at various values of $\varepsilon$ for varying allocation of privacy budget - numerical:categorical on synthetic dataset where only the categorical attributes are correlated with the class.
    }
    \label{cat-cor}
\end{figure}

\begin{figure}
    \centering
    \includegraphics[width=\linewidth, height=6cm]{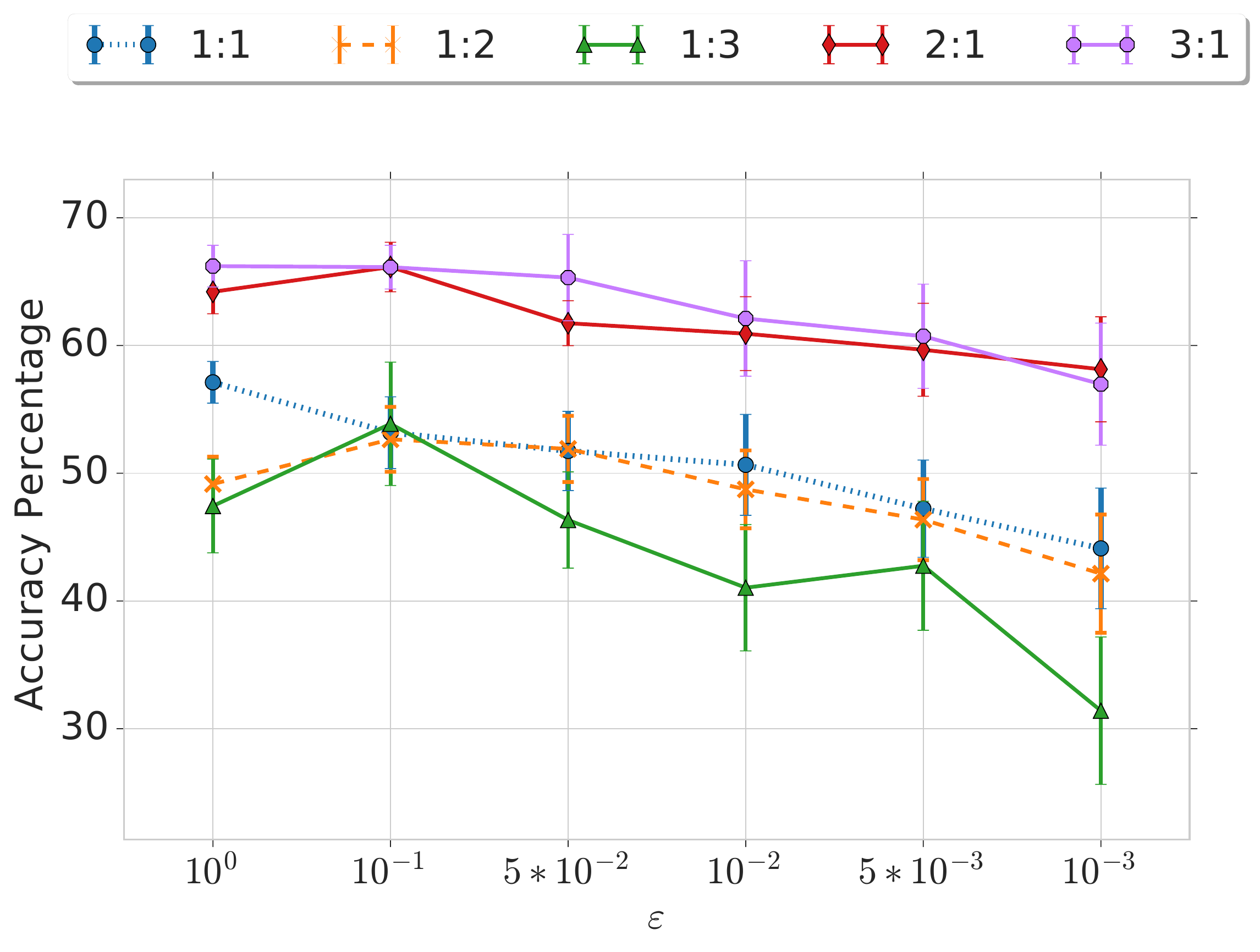}
    \caption{Accuracy at various values of $\varepsilon$ for varying allocation of privacy budget - numerical:categorical on synthetic data where only the numerical attributes are correlated with the class.
    }
    \label{num-cor}
\end{figure}
The unbalanced weight helps numerical more than it hurts categorical, particularly at higher privacy values.  This supports our conjecture that a 2:1 split would be appropriate since we estimate more values for numerical attributes.

\end{document}